\documentclass[twoside,11pt]{article}
\usepackage{jmlr2e}

\jmlrheading{?}{????}{?-??}{?/??}{??/??}{
John Fearnley, Martin Gairing, Paul Goldberg, and Rahul Savani}

\ShortHeadings{Learning Equilibria of Games via Payoff Queries}{Fearnley,
Gairing, Goldberg, and Savani}
\firstpageno{1}

\usepackage{graphicx}
\usepackage{amssymb}
\usepackage{amsmath}

\usepackage{algorithm}
\usepackage{algpseudocode}

\usepackage{tikz}
\usetikzlibrary{automata}

\newcommand{\N}{\ensuremath{N}}
\newcommand{\s}{{\ensuremath{\sf s}}}

\newcommand{\linka}{a}
\newcommand{\x}{\mathbf{x}}
\newcommand{\y}{\mathbf{y}}

\DeclareMathOperator{\query}{Query}

\newcommand{\hide}[1]{}
\newcommand{\player}{p}
\newcommand{\numplayers}{n}
\newcommand{\numstrats}{k}
\newcommand{\numstrat}{\ell}
\newcommand{\degree}{d}
\newcommand{\G}{{\ensuremath {\mathcal G}}}

\newcommand {\xx}{\ensuremath{\delta}}
\newcommand {\bigO}{\mathcal{O}}

\newcommand{\forb}{\textbf{for}}
\newcommand{\qmin}{\ensuremath{q_{min}}}
\newcommand{\qmax}{\ensuremath{q_{max}}}

\renewcommand{\marginparwidth}{60pt}
\usepackage{setspace,color}   

\begin{document}

\title{Learning Equilibria of Games via Payoff Queries}

\author{\name John Fearnley \email John.Fearnley@liverpool.ac.uk \\
       \addr Ashton Building, Ashton Street, University of Liverpool, United Kingdom
		\AND
		\name Martin Gairing \email Gairing@liverpool.ac.uk \\
       \addr Ashton Building, Ashton Street, University of Liverpool, United Kingdom
       \AND
       \name Paul W. Goldberg \email Paul.Goldberg@cs.ox.ac.uk \\
       \addr Wolfson Building, Parks Road, University of Oxford, United Kingdom
		\AND
		\name Rahul Savani \email Rahul.Savani@liverpool.ac.uk \\
       \addr Ashton Building, Ashton Street, University of Liverpool, United
Kingdom}

\editor{}

\maketitle

\begin{abstract}
A recent body of experimental literature has studied {\em empirical
game-theoretical analysis}, in which we have partial knowledge of a game,
consisting of observations of a subset of the pure-strategy profiles and their
associated payoffs to players. The aim is to find an exact or approximate Nash
equilibrium of the game, based on these observations. It is usually assumed
that the strategy profiles may be chosen in an on-line manner by the algorithm.
We study a corresponding computational learning model, and the query complexity
of learning equilibria for various classes of games. We give basic results for
exact equilibria of bimatrix and graphical games. We then study the query
complexity of approximate equilibria in bimatrix games. Finally, we study the
query complexity of exact equilibria in symmetric network congestion games. For
directed acyclic networks, we can learn the cost functions (and hence compute
an equilibrium) while querying just a small fraction of pure-strategy profiles.
For the special case of parallel links, we have the stronger result that an
equilibrium can be identified while only learning a small fraction of the cost
values. 
\end{abstract}



\section{Introduction}

Suppose that we have a game $G$ with a known set of players, and
known strategy sets for each player. We want to design an algorithm to solve $G$,
where the algorithm can only obtain information about $G$ via {\em payoff
queries}. In a payoff query, the algorithm proposes pure strategies for the
players, and is told the resulting payoffs. The general research issue is to
identify bounds on the number of payoff queries needed to find an
equilibrium, subject
to the assumption that $G$ belongs to some given class of games.

\subsection{Motivation}

Given a game, especially one with many players, it is unreasonable to
assume that anyone maintains an explicit representation of its payoff
function, even if the game in question has a concise representation.
However, in practice, a reasonable modelling assumption is that given,
say, a strategy profile for the players, we can determine their payoffs,
or some estimate of the payoffs. We are interested in algorithms that
find Nash equilibria using a sequence of queries, where a query proposes
a strategy profile and gets told the payoffs.
We would like to know under what conditions an algorithm
can find a solution based on knowledge of some but not all of the game's
payoffs, which is particularly important when there
are many players, and the number of pure-strategy profiles is large.
This kind of challenge (where you get observations of
profile/payoff-vector pairs, and you want to find an approximate
equilibrium, as opposed to the unobserved payoffs) has been the subject
of experimental work~\citep{VWS07,Wellman06,JVW08,DVSW09}, where
\cite{JVW08} focuses on the case (highly relevant to this work) where
the algorithm selects a sequence of pure profiles and gets told the resulting payoffs.
{\em In this paper, we introduce the study of payoff-query algorithms from
the algorithmic complexity viewpoint.} We are interested in upper and
lower bounds on the query complexity of classes of games.

From the theoretical perspective, we are studying a constrained class of
algorithms for computing equilibria of games.
The study of such constraints ---especially when they lead to lower bounds
or impossibility results--- informs us about
the approaches that a successful algorithm needs to apply.
In the context of equilibrium computation, other kinds of constraint include
{\em uncoupled} algorithms for computing equilibria~\citep{HMC03,HMC06},
communication-constrained algorithms~\citep{HM10,DFPPV,GP12},
and {\em oblivious} algorithms~\citep{DP09}.
Of course, the restriction to polynomial-time algorithms is the best-known
example of such a constraint.
Based on the algorithms and open problems identified in this paper, we
find this to be a compelling motivation for the further study of
the payoff-query model.
There are various related kinds of query models that are suggested by the payoff
queries studied here, which may also be of similar theoretical interest;
we discuss these in Section~\ref{sec:fw}.

\subsection{Games and query models}
\label{sec:querymodel}

In this paper we introduce the study of payoff-queries for  strategic-form games.
We also consider two models of concisely represented games:  
\emph{graphical games} \citep{KLS01}, where players are nodes in a given graph and the 
payoff of a player only depends on the strategies of its neighbors in the graph, 
and \emph{symmetric network congestion games} \citep{FPT04}, where the strategy space of the players
corresponds to the set of  paths that connect two nodes in a network.

For a strategic-form game, we assume that initially the querying algorithm
only knows~$\numplayers$, the number of players, and $\numstrats$, the
number of pure strategies that each player has. 

\begin{definition}
A {\em payoff query} to a strategic-form game $G$ selects a pure-strategy
profile $\s$ for $G$, and is given as response, the payoffs that $G$'s players
derive from $\s$.
\end{definition}

There are $\numstrats^\numplayers$ pure-strategy profiles in a game, and
one could learn the game exhaustively using this many payoff queries. 
We are interested in algorithms that require only a
small fraction of this trivial upper bound on the number of queries required.

For our results on symmetric network congestion games, we assume that initially
the algorithm only knows the number of players $n$, and the set of pure
strategies, given by a graph and the common origin/destination pair. In this
paper, we will consider two different query models, which are described in
the following definition.

\begin{definition}
\label{def:congest_query}
For a symmetric congestion game with $m$ pure strategies and $n$ players, a
\emph{query} is a tuple $q = (q_1, q_2, \dots, q_m)$, where for each pure
strategy $i=1,2,\ldots,m$, we have that $q_i \in \{0, 1, 2,\ldots, n\}$ is the number of players
assigned to~$i$ under the query.
In response to the query $q$, the querier learns the costs of each pure
strategy under the assigned loads. Let $Q = \sum_{1 \le i \le m} q_i$. We
consider two different types of queries:
\begin{itemize}
\item In a \emph{normal-query}, we require that $Q = n$;
\item in an \emph{under-query}, we require that $Q < n$.
\end{itemize}
\end{definition}

Normal-queries correspond to the query model that we use for strategic-form
games. 
For a congestion game, $m$, which is the number of paths from the origin to 
the destination in a graph, may be exponential.
While  we defined a query for congestion as a tuple of length $m$, both normal-queries
and under-queries require at most $n$ positions of this tuple to be non-zero,
so the query can be specified succinctly.
We use under-queries in our query algorithm for games played on directed
acyclic graphs. We feel that under-queries are a reasonable query model for
congestion games, because we can ask some players to refrain from playing when
we conduct our query. 




\begin{definition}
The {\em payoff query complexity} of a class of games ${\cal G}$, with respect to some
solution concept such as exact or approximate Nash equilibrium, is defined as follows.
It is the smallest $N$ such that there is some algorithm ${\cal A}$ that,
given $N$ payoff queries to any game $G\in{\cal G}$ (where
initially none of the payoffs of $G$ are known) can find a solution of $G$.
\end{definition}

The definition imposes no computational bound on the algorithm ${\cal A}$.
It is to some extent inspired by the work on query-based learning initiated
by~\citet{A87}, in the context of computational learning theory.
Note that ${\cal A}$ may select the queries in an on-line manner, so queries 
can depend on the 
responses to previous queries.

\subsection{Overview of results}

We study a variety of different settings. In Section~\ref{sec:bimatrix}, we
consider bimatrix games. Our first result is a lower bound for computing an
exact Nash equilibrium: in Theorem~\ref{thm:basic}, we show that computing an
exact Nash equilibrium in a $\numstrats\times\numstrats$ bimatrix game has
payoff query complexity $k^2$, even for zero-sum games. In other words, we have
to query every pure strategy profile.

We then turn our attention to approximate Nash equilibria, where we obtain some
more positive results. With the standard assumption that all payoffs lie in the
range~$[0,1]$, we show that, when $2 \le i \le k-1$, the payoff query complexity
of computing a $(1 - \frac{1}{i})$-approximate Nash equilibrium is at most $2k - i
+ 1$ (Theorem~\ref{thm:2playerpayoffs}) and at least $k - i + 1$
(Theorem~\ref{thm:apxlower}.)
We also observe that, when $\epsilon\geq 1-\frac{1}{k}$, no payoff queries are
needed at all, because an $\epsilon$-Nash equilibrium is achieved when
both players mix uniformly over their pure strategies.

The query complexity of computing an approximate Nash equilibrium when $\epsilon
< \frac{1}{2}$ appears to be a challenging problem, and we provide an initial
lower bound in this direction in Theorem~\ref{thm:lowerlog}: we show that the
payoff query complexity of finding a 
$\epsilon$-approximate Nash equilibrium for $\epsilon = \bigO(\frac{1}{\log k})$ is
$\Omega(k \cdot \log k)$. This gives an interesting contrast with the~$\epsilon
\ge \frac{1}{2}$ case. Whereas we can always compute a $\frac{1}{2}$-approximate
with $2k-1$ payoff queries, there exists a 
constant $\epsilon < \frac{1}{2}$ for which this is not the case, as shown 
in~Corollary~\ref{cor:contrast}.

Having studied payoff query complexity in bimatrix games, it is then natural to
look for improved payoff query complexity results in the context of
``structured'' games. In particular, we are interested in {\em concisely
represented} games, where the payoff query complexity may be much smaller than
the number of pure strategy profiles. As an initial result in this direction, in
Section~\ref{sec:gg} we consider graphical games, where we show
(Theorem~\ref{THM:GG}) that for graphical games
with 
constant degree $d$, a Nash equilibrium can be found with a
polynomial number of payoff-queries. This algorithm works by discovering every payoff in the game,
however unlike bimatrix games, this can be done without querying every pure
strategy profile.

Finally, we focus on two different models of congestion games. In
Section~\ref{sec:parallel}, we consider the case of \emph{parallel links},
where the game has a origin and destination vertex, and $m$ parallel links
between them. We show both lower and upper bounds for this setting. If $n$
denotes the number of players, then we obtain a $\log(n) + m$ payoff query lower
bound (Theorem~\ref{thm:LBcongestion}), which applies to both query
models. 
We obtain an upper bound of 
$\bigO\left(\log(n)\cdot \frac{\log^2(m)}{\log\log(m)} + m \right)$ 
normal-queries (Theorem~\ref{thm:UBparallel_n}).
Note that there are $n \cdot m$ different
payoffs in a parallel links game, and so our upper bound implies that you do
not need to discover the entire payoff function in order to solve a parallel
links game.

In Sections~\ref{sec:dags}, \ref{sec:dags2}, \ref{sec:dags3}, we consider the
more general case of symmetric network congestion games on directed acyclic
graphs. We show that if the game has $m$ edges and $n$ players, then we can
find a Nash equilibrium using $m \cdot n$ payoff queries
(Theorem~\ref{thm:UBnetwork}). The algorithm discovers every payoff in the
game, but it only queries a small fraction of the pure strategy profiles.


\section{Related work}

In Section~\ref{sec:pqc} we review some very recent work on the payoff
query complexity of related game-theoretic solution concepts.
In Section~\ref{sec:ew} we review the experimental work that motivated this paper. 
Finally, in Section~\ref{sec:brd} we discuss the relationship with work that analyzes
{\em best-response dynamics} in a game-theretic context.

\subsection{Payoff query complexity}\label{sec:pqc}

A preliminary version of this paper appeared at the ACM conference on Electronic
Commerce~\citep{FGGS13}.
Work that has appeared subsequently has studied query
complexity bounds for general multi-player games, where the main
parameter of interest is the number of players $n$, who usually just
have a small number of pure strategies.
\citet{HN13} obtain an exponential in $n$ lower bound on the query
complexity of finding an exact correlated equilibrium of a general $n$-player game.
Note that any lower bounds for correlated equilibria apply immediately to Nash equilibria,
since Nash equilibria are a more restrictive solution concept.
For {\em approximate} correlated equilibria,
no-regret learning dynamics can be simulated by a randomized payoff query
algorithm, so that the query complexity of
approximate correlated equilibria is polynomial in the number of players~\citep{BB13,HN13}.
\citet{GR13} studied the dependence in more detail, obtaining
upper and lower bounds that are logarathmic in $n$.
However, randomness is needed: \citet{BB13}
show that finding an exact correlated equilibrium in an $n$-player games using
a deterministic querying strategy requires exponentially many queries in $n$.
This result is strengthened by~\citet{HN13}, where it is shown that
deterministic querying strategies require exponentially many queries to find
even a $\frac{1}{2}$-approximate correlated equilibrium.

Approximate \emph{well-supported} Nash equilibria are another approximate
solution concept that have been studied in the context of strategic form
games~\citep{KS10,FGSS12}. \citet{B13} has shown that finding a $10^{-8}$-well
supported Nash equilibrium in an $n$-player game requires exponentially many
queries in~$n$. The query complexity of computing an $\epsilon$-approximate Nash
equilibrium (that need not be well-supported) for constant $\epsilon$ remains
open, although~\citet{GR13} show that it is polynomial if the unknown game can
be specified concisely. These negative results for $n$-player games motivate the
consideration of more structured classes of games, such as congestion games,
which we study in this paper. 

Finally, \citet{FS13} have continued the study of query complexity for bimatrix
games that was initiated in this paper. In particular, they show that randomized
payoff query algorithms can achieve better approximation ratios: there is a
randomized algorithm for finding a $(\frac{3 - \sqrt{5}}{2} + \epsilon)$-Nash
equilibrium in a bimatrix game using $O(\frac{k \cdot \log k}{\epsilon^2})$
payoff queries, and there is a randomized algorithm for finding a $(\frac{2}{3}
+ \epsilon)$-WSNE in a bimatrix game using $O(\frac{k \cdot \log
k}{\epsilon^4})$ payoff queries. They also provide lower bounds for finding
well-supported Nash equilibria in bimatrix games: finding an
$\epsilon$-well-supported Nash equilibrium requires $k - 1$ payoff queries for
any $\epsilon < 1$, even in win-lose games, and finding a
$\frac{1}{3k}$-well-supported Nash equilibrium requires $\Omega(k^2)$ payoff
queries, even in win-lose constant-sum games.

\subsection{Experimental Work}\label{sec:ew}

In {\em Empirical game-theoretic analysis}~(\citet{Wellman06,JSW10}), a game
is presented to the analyst via a set of observations of
strategy profiles (usually, pure) and their corresponding payoffs.
This set of profiles/payoff-vector pairs is called an {\em empirical game}.
In some settings the strategy profiles are randomly generated, but it is typically feasible to obtain
observations via the payoff queries we study here.
The {\em profile selection problem}~\citep{JVW08} is the challenge of choosing
helpful strategy profiles.
The {\em strategy exploration problem}~\citep{JSW10} is the special case
of finding the best way to limit the search to
a small subset of a large set of~strategies. 

\citet{JVW08} envisage a setting where a game (called
a {\em base game}) has a corresponding {\em game simulator},
an implementation in software, which is amenable to
payoff queries; a more general scenario allows the observed payoffs
to be sampled from a distribution associated with the strategy profile.
The distribution is sometimes considered to be due to a noise process, and called the
{\em noisy payoff model} in~\citet{JVW08}.
(In this paper we just consider deterministic payoffs, the ``revealed payoff
model'' in~\citet{JVW08}.)
As noted in~\citet{VWS07}, a profile can be repeatedly
queried to sample from the distribution of payoffs, and thus get an estimate
of the expected values.
The two interacting challenges are to identify helpful queries, and to use
them to find pure-strategy profiles that have low regret
 (where {\em regret} refers to the largest incentive to deviate, amongst the
players.)

\citet{VWS07} study the {\em payoff function approximation task}, 
in which a game belongs to a
known class, and there is a ``\emph{regression}'' challenge to determine 
certain parameters; the information about the game consists of a
random sample of pure profiles and resulting payoff vectors.
However, success is measured by
the extent that the players' predicted behaviour is close to the
behaviour associated with the true payoffs, rather than how well the
true payoff functions are estimated.

Work on specific classes of multi-player games includes the following.
\citet{DVSW09} studies algorithms for learning
graphical games; we consider a graphical game learning algorithm in
Section~\ref{sec:gg}.
\citet{JVW08} apply payoff-query learning to various kinds of games generated by
GAMUT~\citep{NWSL04}, including a class of congestion games. 
\citet{VWS07} investigate a first-price
auction and also a scheduling game, where payoffs are described via a
finite random sample of profile/payoff vector pairs.
Earlier, \citet{SW05} study search for pure Nash
equilibria of strategic-form games (mostly with 5 players and 10 pure strategies).

Most of the experimental work (e.g.,~\citet{SW05,JVW08,DVSW09})
uses 
{\em local search}, in
which profiles that get queried are typically very similar (differing in just one
player's strategy) from previously queried profiles.
Jordan et al.~\citet{JVW08} experiment with local-search type algorithms in which when
a player has the incentive to deviate, the tested profile is updated with that deviation.
\citet{SW05} study search for pure equilibria via best-response
dynamics while maintaining a tabu list, introduced to reduce the risk of cycles.

\subsection{Best-response dynamics and local search}\label{sec:brd}


There is a large body of literature that studies best- and better-response
dynamics for classes of potential games, and gives bounds on the number of steps required
for convergence to pure-strategy equilibria.
These dynamics relate to the payoff query model since they
work by exploring the space of pure profiles, and
receiving feedback consisting of payoffs. The difference is that they purport to
model a decentralized process of selfish behaviour by the players, while the
payoff query model envisages a centralised algorithm that is less constrained.
In this section, we discuss some of the relevant literature.

Local search processes in that each pure
profile is obtained from the previous one by letting a single player move have been 
studied extensively in the literature.
Bounds on the convergence of deterministic best-response dynamics were considered
in \citet{EKM03} and \citet{FGLMR03a}. 
\citet{GS10,GS11} showed polynomial convergence of better-response dynamics for certain \emph{hedonic games}. 
The better-response dynamics considered by \citet{G04} is the basic
randomized local search algorithm, and bounds are obtained for its convergence
to exact equilibrium.
The work in~\citet{BCZ12b} shows that a Nash equilibrium of a bimatrix game can be
found using a polynomial number of better-response queries.
\citet{CS11} study another local search, the {\em $\epsilon$-Nash
dynamics}, and its convergence to approximate equilibria.
\citet{GLMM10} employ controlled local search dynamics (where a sequence of players moves simultaneously)
to computs pure Nash equilibria.
Other papers (e.g.,~\citet{FRV06,BFGGM07}) analyse strongly-distributed dynamics in which
multiple players can move in the same time step; consequently
the dynamics is not a local search. However,
these dynamical systems could all be simulated by payoff query algorithms in which
at each step, at most $nk$ queries are made to determine the change in payoffs available
to players as a result of unilateral deviations. This paper begins to answer the question:
how much better could a payoff query algorithm do, if it were not subject to that
constraint?

Finally, \citet{AEFT11} consider payoff-query algorithms for finding the costs
of paths in graphs. They consider {\em weight discovery protocols} where the aim is to determine
the costs of edges, and {\em shortest path discovery protocols} where the aim is to find
a shortest path. The latter objective is more similar to what we consider, since it can
avoid the need to learn the entire payoff function; also a shortest
path is an equilibrium strategy for the one-player case of a network congestion
game.

\section{Bimatrix games}\label{sec:bimatrix}

In this section, we give bounds on the payoff-query complexity of computing
approximate Nash equilibria of bimatrix games. A \emph{bimatrix game} is a pair
$(R, C)$ of two $k \times k$ matrices: $R$ gives payoffs for the \emph{row
player}, and $C$ gives payoffs for the \emph{column player}. We use $[n]$ to
denote the set $\{1, 2, \dots, n\}$. 
A \emph{mixed strategy} is a probability distribution over $[k]$. A \emph{mixed
strategy profile} is a pair $\s = (\x, \y)$, where $\x$
is a mixed strategy for the row player, and $\y$ is a mixed strategy for the
column player.



Let $\s = (\x, \y)$ be a mixed strategy profile in a $k \times k$ bimatrix game $(R,
C)$. We say that a row $i \in [k]$ is a \emph{best response} for the row
player if $R_i \cdot \y = \max_{j \in [k]} R_j \cdot \y$.  We say that a column
$i \in [k]$ is a best response for the column player if $(\x \cdot C)_i =
\max_{j \in [k]} (\x
\cdot C)_j$. We define the row player's \emph{regret} under $\s = (\x, \y)$ as 
the difference between the payoff of a best response and the payoff that the
row player obtains under $\s$.
More formally, the regret that the row player suffers under $s$ is:
\begin{equation*}
\max_{j \in [k]}( R_j \cdot \y) - \x \cdot R \cdot \y.
\end{equation*}
Similarly, the column player's regret is defined to be:
\begin{equation*}
\max_{j \in [k]}((\x \cdot C)_j) - \x \cdot C \cdot \y.
\end{equation*}
We say that $\s$ is a \emph{mixed Nash equilibrium} if both players have regret
$0$ under $\s$. An \emph{$\epsilon$-Nash equilibrium} is an approximate solution
concept: for every $\epsilon \in [0, 1]$, we say that $\s$ is an $\epsilon$-Nash
equilibrium if both players suffer regret at most $\epsilon$ under $\s$.



We begin with the following simple observation: there are no query-efficient
algorithms for finding \emph{exact} Nash equilibria, even in zero-sum games.
The following theorem shows that, in order to find an exact Nash equilibrium,
we must query all $\numstrats \times \numstrats$ pure strategy profiles.

\begin{theorem}\label{thm:basic}
The payoff query complexity of finding an exact Nash equilibrium of a zero-sum
$\numstrats\times\numstrats$ bimatrix game is $k^2$.
\end{theorem}
\begin{proof}
Consider a generalized version of matching pennies, where the column player
pays~1 to the row player whenever both players choose the same strategy,
otherwise the row player pays~1 to the column player. Note that this is a
zero-sum game, and that it has a unique Nash equilibrium, namely when both
players randomize uniformly over their strategies.
Now suppose each payoff in the game is perturbed by a small quantity, in such a
way as to maintain the zero-sum property.
For small perturbations, there will still be a unique fully-mixed equilibrium
profile,
but it can only be known exactly if all the payoffs are known exactly. Thus, we
cannot find an exact Nash equilibrium in a zero-sum bimatrix game without
querying all $\numstrats \times \numstrats$ pure strategy profiles.
\end{proof}

Theorem~\ref{thm:basic} implies that we cannot devise query-efficient algorithms
for finding exact Nash equilibria. This naturally raises the question of whether
there are query-efficient algorithms for finding \emph{approximate} Nash
equilibria, and we continue by presenting results on this topic. From now on, we will
assume that all payoffs lie in the range $[0,1]$, which is a standard assumption
when finding approximate Nash equilibria.

Our first result is an upper bound. The work of Daskalakis, Mehta, and
Papadimitriou~\citep{DMP} gives a simple algorithm for finding a
$\frac{1}{2}$-Nash equilibrium. We adapt their algorithm to prove the
following result.

\begin{theorem}\label{thm:2playerpayoffs}
Let $i$ be chosen such that $2 \le i \le k - 1$.
The payoff query complexity of finding a $(1 - \frac{1}{i})$-approximate
equilibrium of a $\numstrats\times\numstrats$ bimatrix game is at
most $2k - i + 1$.
\end{theorem}
\begin{proof}
We begin by querying all $\numstrats$ pure profiles where the row player plays
row $1$. This allows us to find the column player's best response to row $1$.
Without loss of generality, we can assume that this is column $1$. Now query
column $1$ against rows $2$ through $k - i + 2$. Note that we have made a total
of $2k - i + 1$ queries. Let row $b$ be a row that maximizes the row player's
payoff against column $1$, among those that we have queried. Let $B = \{1, b\}
\cup [k - i + 3, k]$. We propose the following mixed
strategy profile $\s$: the column player plays column $1$ with probability $1$,
and the row player mixes uniformly over the strategies in $B$. Note that the row
player is mixing between $i$ rows, and thus plays each of them with probability
$\frac{1}{i}$. 

We claim that $\s$ is a $(1 - \frac{1}{i})$-approximate Nash equilibrium. Let
$R$ and $C$ be the actual payoff matrices for the row and column player,
respectively. Note that the row player's best response to column $1$ is either
$b$, or one of the strategies between $k - i + 3$ and $k$. Call this row $j$,
and observe that $j \in B$. The row player's regret can be expressed as
\begin{align*}
R_{j, 1} - \sum_{\ell \in B} \frac{1}{i} \cdot R_{\ell, 1} &= (1 - \frac{1}{i})
\cdot R_{j,
1} - \sum_{\ell \in B \setminus \{j\}} \frac{1}{i} \cdot R_{\ell, 1} \\
& \le (1 - \frac{1}{i}) \cdot R_{j, 1} \\
& \le (1 - \frac{1}{i}).
\end{align*}
Let $j'$ be a pure best response of the column player under $\s$. Observe that,
since column $1$ is a best response against row $1$, we have that $C_{1, j'} -
C_{1, 1} \le 0$. The column player's regret can be expressed as:
\begin{align*}
\sum_{\ell \in B} \frac{1}{i} \cdot C_{\ell, j'} - \sum_{\ell \in B} \frac{1}{i}
\cdot C_{\ell, 1} &= \sum_{\ell \in B} \frac{1}{i} \cdot (C_{\ell, j'} - C_{\ell, 1}) \\
& \le \sum_{\ell \in B \setminus \{1\}} \frac{1}{i} \cdot (C_{\ell, j'} - C_{\ell, 1}) \\
& \le \sum_{\ell \in B \setminus \{1\}} \frac{1}{i} \\
&= 1 - \frac{1}{i}.
\end{align*}
Thus we have shown that both players suffer regret at most $1-\frac{1}{i}$.
\end{proof}

Note that, when $i = 2$, the algorithm of Theorem~\ref{thm:2playerpayoffs}
finds a $\frac{1}{2}$-Nash equilibrium using the same technique as the algorithm
from~\citet{DMP}. For $i > 2$, our algorithm uses fewer payoff queries in
exchange for a worse approximation. When $i = k-1$, our algorithm uses $k+2$
payoff queries in order to find a $(1 - \frac{1}{k-1})$-Nash equilibrium. It
turns out that, for $\epsilon\geq 1-\frac{1}{k}$, we do not need to make any
payoff queries at all: an $\epsilon$-Nash equilibrium is obtained when both
players play the uniform distribution over their strategies, because both
players must place at least $\frac{1}{k}$ of their probability on a pure best
response.

We now turn our attention to lower bounds. We complement the
result of Theorem~\ref{thm:2playerpayoffs} by showing lower bounds for finding
$(1 - \frac{1}{i})$-Nash equilibria, when $i$ is in the range $2 \le i \le k -
1$. First, we prove an auxiliary lemma.


\begin{lemma}
\label{lem:noquery}
Suppose that all payoff queries return $0$ for both players.
Let $i$ be chosen such that $2 \le i \le k-1$, and let $\s$ be a
$(1-\frac{1}{i})$-Nash equilibrium. Any column that receives no queries
must be assigned at least $\frac{1}{i}$ probability by $\s$.
\end{lemma}
\begin{proof}
Suppose, for the sake of contradiction, that $c$ is a column that received no
queries, and that $c$ is assigned strictly less than $\frac{1}{i}$ probability
by $\s$. We construct a column player matrix $C$ as follows: 
\begin{equation*}
C_{j, j'} = \begin{cases}
1 & \text{if $j' = c$,}\\
0 & \text{otherwise.}
\end{cases}
\end{equation*}
Since~$c$ received no queries, $C$ is consistent with all queries that have been
made. Note that the column player's payoff under $\s$ is strictly less than
$\frac{1}{i}$, and that the payoff of playing $c$ as a pure strategy is $1$.
Thus, the column player's regret is strictly greater than $1 - \frac{1}{i}$,
which contradicts the fact that $\s$ is a $(1-\frac{1}{i})$-Nash equilibrium
\end{proof}

Now we can show our lower bound.

\begin{theorem}
\label{thm:apxlower}
Let $i$ be chosen such that $2 \le i \le k-1$. 
The payoff query complexity of finding a
$(1-\frac{1}{i})$-approximate Nash equilibrium of a 
$\numstrats\times\numstrats$ bimatrix game is at least $k-i + 1$.
\end{theorem}
\begin{proof}
Assume that all payoff queries return $0$ for both players.
Suppose, for the sake of contradiction, that an algorithm makes fewer than
$k-i+1$
payoff queries, and then outputs~$\s$ as a $(1-\frac{1}{i})$-Nash equilibrium.
It follows that there must be at least $i$ columns that have received no payoff
queries at all,
 and without loss of generality, we can
assume that these are columns $1$ through $i$. By Lemma~\ref{lem:noquery}, we
know that $\s$ must assign exactly $\frac{1}{i}$ probability to each of the
columns $1$ through $i$. Since there are $k$ rows, there is at least one row $r$
that receives probability at most $\frac{1}{k}$ under $\s$. We construct a row
player payoff matrix~$R$ as follows:
\begin{equation*}
R_{j,j'} = \begin{cases}
1 & \text{if $j = r$ and $1 \le j' \le i$,}\\
0 & \text{otherwise.}
\end{cases}
\end{equation*}
Since columns $1$ through $i$ were not queried, $R$ is consistent with all
queries that have been made so far. The row player's payoff under $\s$ is at
most $\frac{1}{k}$. On the other hand, the row player would receive payoff $1$
for playing $r$ as a pure strategy. Thus, the row player's regret is at least:
\begin{equation*}
1 - \frac{1}{k} > 1 - \frac{1}{i}.
\end{equation*}
This contradicts the fact that $\s$ is a $(1 - \frac{1}{i})$-Nash equilibrium.
\end{proof}

As a consequence of the previous two theorems, when $2 \le i \le k-1$, we have
that the payoff query complexity of finding a $(1 - \frac{1}{i})$-Nash
equilibrium lies somewhere in the range $[k-i+1, 2k - i +1]$. Determining the
precise payoff query complexity for this case is an open problem.

So far, we have only considered $\epsilon$-Nash equilibria with $\epsilon \ge
\frac{1}{2}$. Of course, the most interesting challenge is to determine the
payoff query complexity for values of $\epsilon<\frac{1}{2}$. By our previous
results, we know that the payoff query complexity for finding a
$\frac{1}{2}$-Nash equilibrium is $\bigO(k)$, and the payoff query complexity for
finding a $0$-Nash equilibrium is $\bigO(k^2)$, but we do not know how the payoff
query complexity behaves as we vary $\epsilon$ between $0$ and $\frac{1}{2}$.

Our final result in this section will be to show a lower bound for $\epsilon =
\bigO(\frac{1}{\log k})$. We will show that finding a $\bigO(\frac{1}{\log k})$-Nash
equilibrium requires $\Omega(k \log k)$ payoff queries. This establishes that
there are some positive values of $\epsilon$, for which computing an
$\epsilon$-Nash equilibrium is asymptotically harder than computing a
$\frac{1}{2}$-Nash equilibrium.

We will use the following class of bimatrix games, which have been previously
used in Theorem 1 of~\citet{FNS07}.

\begin{definition}
Let $\G_\numstrat$ be the class of strategic-form games where the column player
has $\numstrat$ pure strategies and the row player has $\binom{\numstrat}{\numstrat/2}$ pure
strategies (where we assume $\numstrat$ is even).
Let $G_\numstrat\in \G_\numstrat$ be the win-lose constant-sum game in
which each row of the row
player's payoff matrix has $\frac{\numstrat}{2}$ 1's and $\frac{\numstrat}{2}$ 0's, all rows being
distinct. The column player's payoffs are one minus the row player's payoffs.
\end{definition}

It is well-known that every zero-sum game has a unique \emph{value}, which is the
payoff that both players can guarantee for themselves, independent of what the
other player does.
The value of each game $G_\numstrat \in \G_\numstrat$ is $\frac{1}{2}$
since either player can obtain payoff $\frac{1}{2}$ by using the uniform
distribution over their pure strategies. Our first lemma shows that, if the
column player deviates from this by placing too much probability on a single
column, then the row player can take advantage and increase his payoff.

\begin{lemma}\label{lem:mixed}
Suppose that in game $G_\numstrat \in \G_\numstrat$, the column player places
probability $\alpha>1/\numstrat$ on some column. Then the row player can obtain
a payoff strictly greater than
$\frac{1}{2}+\frac{\alpha}{2}-\frac{1}{2\numstrat}$.
\end{lemma}

\begin{proof}
Let $j$ be a column that the column player plays with probability $\alpha$.
Let $R_j$ be the set of rows where the row player obtains payoff 1 against column $j$.
Suppose the row player plays the uniform distribution over rows in $R_j$.
When the column player plays $j$, the row player receives payoff 1.
Let $j'\not=j$ be a column, and consider the payoffs to the row player where $j'$ intersects $R_j$.
A fraction $\frac{\numstrat/2-1}{\numstrat-1}$ of these entries pay the row player 1, while a fraction
$\frac{\numstrat/2}{\numstrat-1}$ pay the row player 0.
Consequently whenever the column player plays $j'\not=j$, the row player's expected
payoff is $\frac{\numstrat/2-1}{\numstrat-1}$.
Thus with probability $\alpha$ the row player receives payoff~1, and with probability
$1-\alpha$ he receives payoff $\frac{\numstrat/2-1}{\numstrat-1}$. Thus, the
payoff to the row player is
%
\begin{align*}
\alpha + (1-\alpha)\frac{\numstrat/2-1}{\numstrat-1} &=
\frac{1}{2}+\frac{1}{2}\alpha - \frac{1-\alpha}{2(\numstrat-1)} \\
& >\frac{1}{2}+\frac{1}{2}\alpha - \frac{1-1/\numstrat}{2(\numstrat-1)} \\
& =\frac{1}{2}+\frac{1}{2}\alpha - \frac{1}{2\numstrat}\ ,
\end{align*}
which completes the proof.
\end{proof}

We now use the bound from the previous lemma to show that, in an approximate
Nash equilibrium for $G_\ell$, the column player cannot place too much
probability on any individual column.

\begin{corollary}\label{cor:alpha}
Let $\alpha>\frac{1}{k}$, and let $\epsilon = \frac{1}{4}
(\alpha-\frac{1}{\numstrat})$. In every $\epsilon$-Nash equilibrium of
$G_\numstrat \in \G_\numstrat$, the column player plays each individual column with probability
at most $\alpha$.
\end{corollary}

\begin{proof}
Suppose, for the sake of contradiction, that there is an $\epsilon$-Nash
equilibrium $\s$ in which that the column player assigns column~$j$ probability
strictly greater than $\alpha$. Then, by Lemma~\ref{lem:mixed}, the row player's
payoff is strictly greater than
$\frac{1}{2}+\frac{\alpha}{2}-\frac{1}{2\numstrat}$, and therefore the row
player's payoff in $\s$ must be strictly greater than:
\begin{equation*}
\frac{1}{2}+\frac{\alpha}{2}-\frac{1}{2\numstrat}-\epsilon
=\frac{1}{2}+\epsilon.
\end{equation*}
Therefore, the column player obtains payoff strictly less than
$\frac{1}{2}-\epsilon$. Since the value of $G_\numstrat$ is~$\frac{1}{2}$, the
column player's regret in $\s$ is strictly greater than $\epsilon$, and
therefore $\s$ is not an $\epsilon$-Nash equilibrium.
\end{proof}

We can now provide a lower bound for the payoff query complexity of finding an
approximate Nash equilibrium for the games in $\G_\numstrat$.

\begin{lemma}\label{thm:epsne}
For any $\epsilon<\frac{1}{12}$, and any even $\numstrat\geq 8$, the payoff
query complexity of finding an $\epsilon$-Nash equilibrium for the games in
$\G_\numstrat$ is at least $\frac{1}{2} \cdot
\binom{\numstrat}{\numstrat/2}\cdot (\frac{1}{16\epsilon+4/\numstrat})$.
\end{lemma}

\begin{proof}
Let ${\cal A}$ be a payoff query algorithm for finding an $\epsilon$-Nash
equilibrium, and, for the sake of contradiction, suppose that $\mathcal{A}$ makes
fewer than $\frac{1}{2} \cdot \binom{\numstrat}{\numstrat/2}\cdot
(\frac{1}{16\epsilon+4/\numstrat})$ many payoff queries when processing
$G_\numstrat$.
Let $\s$ be the mixed strategy profile that ${\cal A}$ outputs for
$G_\numstrat$. By Corollary~\ref{cor:alpha}, we know that no column in $\s$ is
assigned more than $\alpha=4\epsilon+\frac{1}{\numstrat}$ probability. We also
know that in $\s$, the row player's payoff is at most $\frac{1}{2}+\epsilon$,
since $\s$ is an $\epsilon$-Nash equilibrium of a constant-sum game with value
$\frac{1}{2}$. Since ${\cal A}$ made fewer than
$\frac{1}{2} \cdot \binom{\numstrat}{\numstrat/2}\cdot
(\frac{1}{16\epsilon+4/\numstrat})$ payoff queries, at least half of the rows
received fewer than $(\frac{1}{16\epsilon+4/\numstrat})$ queries. Since
$\numstrat \geq 8$, this implies that there are at least $\frac{1}{2} \cdot
\binom{8}{4} = 45$ such rows. Thus, there is one such row, call it $r$, that is
played with probabilty strictly less than $\frac{1}{12}$ in $\s$.

Since $\s$ assigns at most $\alpha$ probability to each column, the total amount
of probability that $\s$ assigns to the queried portion of $r$ is at most
$\alpha(\frac{1}{16\epsilon+4/\numstrat}) = \frac{1}{4}$. Now suppose that we
modify $G_\numstrat$ by replacing all un-queried entries of $r$ with payoffs of
1 for the row player. Call this new game $G'_\numstrat$. Note that ${\cal A}$
outputs the same strategy profile $\s$ for both $G_\numstrat$ and
$G'_\numstrat$.

Let $p$ be the payoff to the row player of playing $\s$ in $G_\numstrat$, and
let $p'$ be the payoff to the row player of playing $\s$ in $G'_\numstrat$.
Since $r$ is played with probability less than $\frac{1}{12}$ we have:
\begin{align*}
p'  &\le p + \frac{1}{12} \\
&\le \frac{7}{12} + \epsilon
\end{align*}
However, the row player's best response payoff is at least $\frac{3}{4}$ in
$G'_\numstrat$, so we have:
\begin{equation*}
p' \ge \frac{3}{4} - \epsilon
\end{equation*}
Therefore, we can conclude that:
\begin{align*}
\frac{7}{12} + \epsilon & \ge \frac{3}{4} - \epsilon \\
2 \epsilon &\ge \frac{2}{12}.
\end{align*}
However, this is impossible because $\epsilon < \frac{1}{12}$.
\end{proof}

Finally, we can extend the lower bound to square bimatrix games.

\begin{lemma}\label{lem:main}
For $\numstrats\times\numstrats$ bimatrix games, the payoff query
complexity of finding an $\epsilon$-Nash equilibrium, for
$\epsilon\leq\frac{1}{8}$, is at least
${\numstrats}\cdot(\frac{1}{32/\log\numstrats+64\epsilon})$.
\end{lemma}

\begin{proof}
Let $\numstrats'$ be the largest number of the form $\binom{\numstrat}{\numstrat/2}$
that is smaller than $\numstrats$.
We have $\numstrats'\geq \numstrats/4$ and $\numstrat\geq \log\numstrats/2$.
By Lemma~\ref{thm:epsne},
the number of payoff queries needed to find an $\epsilon$-Nash equilibrium for
games in
$\G_{\numstrats}$ is at least:
\begin{align*}
\binom{\numstrat}{\numstrat/2}\cdot \Bigl(\frac{1}{16\epsilon+4/\numstrat}\Bigr) &= 
\numstrats'\Bigl(\frac{1}{4/\numstrat+16\epsilon}\Bigr) \\
&\geq \frac{\numstrats}{4} \Bigl(\frac{1}{4/\numstrat+16\epsilon}\Bigr) \\
&\geq \frac{\numstrats}{4} \Bigl(\frac{1}{8/\log(\numstrats)+16\epsilon}\Bigr) \\
&= \numstrats \Bigl(\frac{1}{32/\log(\numstrats)+64\epsilon}\Bigr).
\end{align*}
The games in $\G_{\numstrat}$ can be written down as a $\numstrats\times\numstrats$
game, by duplicating rows and columns. Note that these operations preserve
approximate equilibria. 
\end{proof}

By taking $\epsilon \in \bigO(\frac{1}{\log k})$ in the previous lemma, we arrive at
our final theorem.

\begin{theorem}
\label{thm:lowerlog}
For $\numstrats\times\numstrats$ bimatrix games, the payoff query
complexity of finding a $\epsilon$-Nash equilibrium for $\epsilon \in
\bigO(\frac{1}{\log k})$, is $\Omega(k \cdot \log k)$.
\end{theorem}

Recall, from Theorem~\ref{thm:2playerpayoffs}, that we can always find a
$\frac{1}{2}$-Nash equilibrium using $2k - 1$ payoff queries. The following
corollary of Lemma~\ref{lem:main} shows that there are some constant values of
$\epsilon$ that require more payoff queries.

\begin{corollary}
\label{cor:contrast}
There is a constant value of $\epsilon > 0$ for which finding an
$\epsilon$-Nash equilibrium of a $\numstrats\times\numstrats$ bimatrix game
requires strictly more than $2k - 1$ payoff queries.
\end{corollary}

\begin{proof}
Consider, for example, setting
$\epsilon = \frac{1}{512}$ in Lemma~\ref{lem:main}.
Then, for the family
of games in $\G_l$ with $l > 2^{256}$, we have a lower bound
of 
\begin{equation*}
k \cdot \left( \frac{1}{\frac{32}{\log k} + 0.0064} \right) > k \cdot
\frac{1}{0.125 + 0.125}
= 4 \cdot k,
\end{equation*} 
on the number of payoff queries.
\end{proof}

An interesting question that remains is whether one can a show a superlinear
lower bound on the number of payoff queries required for a constant $\epsilon$.

\section{Graphical games}\label{sec:gg}

In this section, we give a simple payoff query-based algorithm for graphical
games. In a $n$-player {\em graphical game}~\citep{KLS01} the players lie at the
vertices of a degree-$\degree$ graph, and a player's payoff is a function of the
strategies of just himself and his neighbors. If every player has $k$ pure
strategies, then the number of payoff values needed
to specify such a game is $\numplayers\cdot\numstrats^{\degree+1}$ which, in
contrast with strategic-form games, is polynomial (assuming $\degree$ is a
constant).

Previously, \citet{DVSW09} have carried out experimental work on
payoff queries for graphical games. They compare a number of techniques; the
algorithm we give here is polynomial-time but would likely be less efficient in
practice. Similar to~\citet{DVSW09}, we assume the underlying graph $G$ is
unknown, and we want to induce the structure of $G$, and corresponding payoffs.

\begin{theorem}\label{THM:GG}
For constant $d$, the payoff query complexity of degree $\degree$ graphical games
is polynomial. 
\end{theorem}
\begin{proof}
Algorithm~\ref{alg:gg} constructs a directed graph $G$ for
the (initally unknown) game, along with the payoff function.
$G$ is the ``affects graph''~\citep{GP06} in which a directed edge $(p',p)$ has the
meaning that the behaviour of $p'$ may affect $p$'s payoff.
Note that in Step~\ref{step-S}, $|S| < (n\cdot\numstrats)^{d+1}$.
In a degree-$\degree$ graphical game, any player $p$'s payoffs may be affected
by his own strategy, and the strategies of at most $d$ neighbours $p'$
for which edges $(p',p)$ exist.
The existence of edge $(p',p)$ is equivalent to the
existence of strategy profiles $\s$, $\s'$ that differ only in $p'$'s
strategy and $p$'s payoff. This is what Algorithm~\ref{alg:gg}
checks for.
Finally, when the edges, and hence neighborhoods of the graph game have been
found, it is simple to read off each player's payoff matrix from the data
in Step~\ref{querystep}.
\end{proof}

\begin{algorithm}
\begin{algorithmic}[1]
\State Initialize graph $G$'s vertices to be the player set, with no edges
\State\label{step-S}{Let $S$ be the set of pure profiles in which at
          least $n-(d+1)$ players play 1.}
\State\label{querystep}{Query each element of $S$.}
\ForAll {players $\player$, $\player'$}{
  \If {$\exists\s,\s' \in S$ that differ only in $p$'s payoff and $p'$'s strategy}
  \State add directed edge $(p,p')$ to graph
\EndIf
}
\EndFor
\ForAll {players $p$}
\State Let $N_p$ be $p$'s neighborhood in $G$
\State Use elements of $S$ to find $p$'s payoffs as a function of strategies of $N_p$
\EndFor
\end{algorithmic}
\caption{\textsc{GraphicalGames}}\label{alg:gg}
\end{algorithm}

Algorithm~\ref{alg:gg} learns the entire payoff function with polynomially many
queries, but there are a couple of important caveats. First, although the payoff
query complexity is polynomial, the computational complexity is probably not
polynomial, since it is PPAD-complete to actually compute an approximate Nash
equilibrium for graphical games~\citep{DGP}. Second, while
Algorithm~\ref{alg:gg} avoids querying all of the exponentially-many
pure-strategy profiles, it works in a brute-force manner that learns the entire
payoff function. It is natural to prefer algorithms that find a solution without
learning the entire game, such as those that we give for
Theorem~\ref{thm:2playerpayoffs} and Theorem~\ref{thm:UBparallel_n}.

\section{Congestion games}
\label{sec:congest}


In this section, we give bounds on the payoff-query complexity of
finding a pure Nash equilibrium in symmetric network congestion games.
A congestion game is defined by a tuple $\Gamma=(\N, E, (S_i)_{i\in\N},
(f_e)_{e\in E})$. Here, $\N=\{1, 2, \ldots, n\}$ is a set of $n$ players and
$E$ is a set of resources. Each player chooses  as her \emph{strategy} a set
$s_i\subseteq E$ from a given \emph{set of available strategies} $S_i\subseteq
2^E$. 
Associated with each resource $e\in E$ is a non-negative, non-decreasing
function $f_e: \mathbb{N}\mapsto \mathbb{R}^+$. These functions describe
\emph{costs} (latencies) to be charged to the players for using resource~$e$.
An outcome (or strategy profile) is a choice of strategies $\s = (s_1, s_2, ..., s_{n})$ by
players with $s_i \in S_i$. For an outcome $\s$ define
$n_e(\s)=|i\in\N: e\in s_i|$ as
the number of players that use resource $e$. The \emph{cost} for
player $i$ is defined by $c_i(\s)=\sum_{e\in s_i} f_e(n_e(\s))$.
%
 A \emph{pure Nash equilibrium} is an outcome $\s$ where no player has an
incentive to deviate from her current strategy. Formally, $\s$ is a pure Nash equilibrium if
for each player $i\in\N$ and $s_i' \in S_i$, which is an alternative strategy for
player $i$, we have $c_i(\s) \leq c_i(\s_{-i}, s_i')$.
Here $(\s_{-i}, s_i')$ denotes the outcome that results when player $i$ changes
her strategy in $\s$ from $s_i$ to $s_i'$.

In a \emph{network congestion game}, resources correspond to the edges in a
directed multigraph $G=(V,E)$. Each player $i$ is assigned an origin node $o_i$,
and a destination node $d_i$.
A strategy for player~$i$ consists of a sequence of edges that form a directed
path from $o_i$ to $d_i$, and the strategy set $S_i$ consists of all such paths.
In a \emph{symmetric} network congestion game all players have the same origin
and destination nodes. We write a symmetric network congestion game as
$\Gamma=(\N, V, E, (f_e)_{e\in E}, o, d)$, where collectively~$V$,~$E$,~$o$, and~$d$
succinctly define the strategy space $(S_i)_{i\in\N}$.
We consider two types of network, directed acyclic graphs, and the special case
of parallel links.
We assume that initially we only know the number of players $n$ and the
strategy space.
The latency functions are completely unknown initially.
As discussed in Section~\ref{sec:querymodel}, we use several different 
querying models for congestion games.

\subsection{Parallel links}
\label{sec:parallel}

In this section, we consider congestion games on $m$ parallel links.  
We present a lower bound and an upper bound on the query complexity of 
finding an exact pure equilibrium of these games.
To simplify the presentation of the algorithmic ideas of our upper bound we
introduce a stronger type of query that we call an \emph{over-query}.
Recall from Definition~\ref{def:congest_query} that for a query $q = (q_1, q_2, \dots, q_m)$, 
we denote by $Q$ the total number of players used in the query, i.e., $Q =\sum_{1
\le i \le m} q_i$.

\begin{definition}
\label{def:congest_over-query}
An \emph{over-query} is a query with $n < Q \le mn$.
\end{definition}

First, we present a simple lower bound. Then, we present an algorithm,
Algorithm~\ref{fig_alg1}, that uses over-queries. Finally, we extend
Algorithm~\ref{fig_alg1} to Algorithm~\ref{fig_alg1b}, which uses only normal
queries.

\paragraph{Lower bound.} In the following construction, we
show that, if there are two links, the querier can do no better than
performing binary search in order to find an equilibrium, which gives a lower
bound of $\log(n)$ many queries. 

\begin{theorem}\label{thm:LBcongestion}
	A querier must make $\log(n)$ queries to determine a pure
equilibrium of a symmetric network congestion game played on parallel links.
\end{theorem}
\begin{proof}
We fix a graph $G$ with two parallel links $e_1$ and $e_2$, and we fix the
cost of $e_2$ so that $f_{e_2}(i) = 1$ for all $i \in N$. We consider functions
$f_{e_1}$ that only return costs of $0$ or $2$. Since $f_{e_1}$ is
non-decreasing, this implies that it will be a step function with a single step.
We say that the step is at location $i \in N$ if $f_{e_1}(j) = 0$ for all $j \le
i$, and $f_{e_1}(j) = 2$ for all $j > i$. The precise location of the step will
be decided by an adversary, in response to the queries that are received. 

The adversary's strategy maintains two integers $\ell$ and $u$ with $\ell < u$, and
initially the adversary sets $\ell = 0$ and $u = n$. Intuitively, for all values
below $\ell$ the adversary has fixed $f_{e_1}$ to $0$, and for all values above $u$
the adversary has fixed $f_{e_1}$ to $2$. The range of values between $u$ and
$\ell$ are yet to be fixed, and all values in this range could potentially be the
location of the step. 

Suppose that the adversary receives the query~$\s$. The adversary will respond
with a pair $(c_1, c_2)$, where~$c_1$ is the cost of~$e_1$, and~$c_2$ is the
cost of~$e_2$. The adversary uses the following strategy: 
\begin{itemize}
\item If $n_{e_1}(\s) \le \ell$, then the adversary responds with $(0, 1)$. If
$n_{e_1}(\s) \ge u$, then the adversary responds with $(2, 1)$.
\item If $n_{e_1}(\s) < \frac{u + \ell}{2}$, that is, if $n_{e_1}(\s)$ is closer to
$\ell$ than it is to $u$, then the adversary sets $\ell = n_{e_1}(\s)$, and responds
with $(0, 1)$.
\item If $n_{e_1}(\s) \ge \frac{u + \ell}{2}$, that is, if $n_{e_1}(\s)$ is closer
to $u$ than it is to $\ell$, then the adversary sets $u = n_{e_1}(\s)$, and
responds with $(2, 1)$.
\end{itemize}

Note that, if there exists an $i$ with $\ell < i < u$, then the querier cannot
correctly determine the Nash equilibrium. This is because the step could be at
location $i$, or it could be at location $i-1$. In the former case, the unique
Nash equilibrium assigns $i$ players to $e_1$ and $n-i$ players to $e_2$, and in
the latter case the unique Nash equilibrium assigns $i-1$ players to $e_1$ and
$n-i+1$ players to $e_2$. By construction, the adversary's strategy ensures
that, in response to each query, the gap between $u$ and $\ell$ may decrease by at
most one half. Thus, the querier must make $\log(n)$ queries to correctly
determine the Nash equilibrium.
\end{proof}

Consider a one-player game with $m$ links.
Clearly, we can solve this game with a single over-query, but 
it requires $m$ normal-queries.
Thus we have the following:

\begin{corollary}
\label{cor:LBlogn+m}
If over-queries are not allowed, then $\log(n) + m$
queries are required to determine a pure equilibrium of a symmetric network
congestion game played on parallel links. 
\end{corollary}

\paragraph{Upper bound.} In the rest of the section, we provide an upper bound,
by constructing a payoff query algorithm that finds a pure Nash equilibrium
using $\bigO\left(\log(n)\cdot \frac{\log^2(m)}{\log\log(m)} + m \right)$
normal-queries.
In order to simplify the presentation, we first present an algorithm that makes
use of over-queries; later we show how this can be translated into an algorithm
that uses only normal-queries.

\paragraph{\bf Overview of algorithm with over-queries.} 
Our algorithm is based on an algorithm from \citet{GLMMR08}. Before we present
the full algorithm, we give an overview of the techniques by describing
a simplified version of the algorithm. The basic idea is to group the players
into blocks, where all players in a block must play on the same link. In each
round of the algorithm, we maintain the property that the blocks are in
equilibrium: no block of players can collectively deviate in order to reduce
their latency. Initially, we place all of the players into a single block, and
then in each round of the algorithm, we split each block into smaller blocks,
and compute a new equilibrium for the smaller block size. Eventually, the block
size will be reduced to $1$, and we recover a Nash equilibrium for the
congestion game.

In this simplified overview, we will assume that the number of players $n$ is
equal to $2^i$ for some $i \in \mathbb{N}$, and in each round we will split each
block in half. Our full algorithm will be more complicated, because it
must deal with an arbitrary number of players, and it will split each block into
more than two pieces.

At the start of the algorithm, we place all $n$ players into a single block. In
order to find an equilibrium for this block, we simply have to find the link $i
\in [m]$ that minimizes $f_i(n)$. We can do this with a single over-query $q =
(n, n, \dots, n)$. 

Now suppose that we have found an equilibrium $\s$ for block size $\delta$. We
split each block into two equal-sized pieces, and our task is to transform $\s$
into an equilibrium for block size $\delta/2$ by moving blocks between the
links. The key observation is that no link can receive two or more blocks
of size $\delta/2$, because this would contradict the fact that $\s$ is an
equilibrium for block size $\delta$. So, when we move blocks between the links,
we know that each link can receive at most one block, and therefore each link
can lose at most $m - 1$ blocks. We can make a single over-query in order to
discover the cost of adding one block of $\delta/2$ players to each link: we
simply query $p = (n_{1}(\s) + \delta/2, n_2(\s) + \delta/2, \dots, n_{m}(\s) +
\delta/2)$.  On the other hand, we also need to determine how many blocks each
link loses, and a naive approach would use $m$ queries. We now describe a method
that uses only $\log^2(m)$ under-queries.

Suppose that we guess that $q$, where $0 \le q \le m$, is the number of blocks
that move. We give an algorithm that verifies whether this guess is correct. Let
$c$ be the $(q+1)$th smallest cost returned by the query $p$. For each link $i$,
we determine $q_i$, which is the number of $\delta/2$-sized blocks that would
want to move to a link with cost $c$. This can be done by binary search, in
parallel for all links, using $\log(m)$ many under-queries. There are three
possible outcomes:
\begin{itemize}
\item If $\sum_{i = 1}^m q_i = q$, then our guess was correct, and exactly $q$
blocks move. 
\item If $\sum_{i=1}^m q_i < q$, then our guess was too high, and fewer than $q$
blocks move.
\item If $\sum_{i=1}^m q_i > q$, then our guess was too low, and more than $q$
blocks move.
\end{itemize}
Thus, to determine exactly how many blocks move between the links, we can use a
nested binary search approach: in the outer level we guess how many blocks move,
and in the inner level we use the above method to determine if our guess was too
high or too low.

Therefore, we have a method for constructing an equilibrium with block size
$\delta/2$ from an equilibrium with block size $\delta$ using $\log^2(m)$ many
queries. Since we start with block size $n$, and we halve the block size in
every round, this gives us an algorithm that finds a Nash equilibrium using
$\log(n) \cdot \log^2(m)$ many payoff queries.

In the rest of this section, we formalize this approach, and we deal with the
issues that were ignored in this high level overview. In particular, we present
an algorithm that works for any number of players $n$, and we obtain a slightly
better query complexity by splitting each block into $\log(m)$ many pieces in
each round.

\paragraph{\bf The algorithm with over-queries.}

The algorithm \textsc{ParallelLinks} is depicted in Algorithm \ref{fig_alg1}.
We will show how this algorithm can be implemented with 
$\bigO\left(\log(n)\cdot \frac{\log^2(m)}{\log\log(m)}\right)$ queries. 
The integer $k$ is a parameter to the algorithm that determines the
block size: in each round we consider blocks of size $k^t$ for some $t$. To
deal with the fact that $n$ may not be an exact power of $k$, the algorithm
will maintain a special link $\linka$. This link is defined to be
the link upon which all $n$ players are placed at the start of the algorithm.
Since every subsequent step of the algorithm only moves players in blocks of
size $k^t$ for some $t$, link $\linka$ will be the only like where the number
of players is not a multiple of the block size.

We start by formalising the notion of an equilibrium with respect to a certain
block size. For a congestion game $\Gamma$, an integer $\xx$, and a special link
$\linka$ we define a $\xx$-equilibrium as follows: 
\begin{definition}[$\xx$-equilibrium]
\label{d:deltaequilibrium}
A strategy profile $\s$ is $\xx$-equilibrium if $\xx\vert n_i(\s)$ for all
$i\in[m]\setminus\{a\}$, and for all links $i,j\in[m]$ with $n_i(\s) \ge \xx$ we
have $f_i(n_i(\s))\le f_j(n_j(\s) + \xx)$.
\end{definition}


Intuitively, we can think of a $\xx$-equilibrium $\s$ as a Nash equilibrium in a
transformed game where the players (of the original game) are partitioned into
blocks of size $\xx$ and each block represents a player in the transformed game,
and the remaining $(n \bmod \xx)$ players are fixed to link $\linka$.

\begin{algorithm}
\algblock[Parallel]{Parallel}{EndParallel}
\algblockdefx[PARALLEL]{PARALLEL}{ENDPARALLEL}%
    [1][Unknown]{Parallel~{}#1}%
    {EndParallel}
\begin{algorithmic}[1]
\State $a\gets \arg\min_{i\in[m]} f_i(n)$
\Comment{1 over-query}
\label{l:bestlink}
\State initialize strategy profile $\s$ by putting all players on link $a$ 
\State $T \gets \lfloor\frac{\log(n)}{\log(k)}\rfloor$
\For{$t= T, T-1, \ldots, 1, 0$}
\State $\delta \gets k^t$  \label{l:delta}
\State $\s \gets$ \textsc{RefineProfile($\s,\xx,0,km$)}  
\EndFor
\State
\Return \s
\medskip
\hrule
\medskip
\Function{RefineProfile}{$\s,\xx,\qmin,\qmax$}
\State $q\gets \lfloor\frac{\qmin+\qmax}{2}\rfloor$
\smallskip
\Parallel{ for all links $i\in[m]$} \label{l11}
\State Query for costs $f_i(n_i(\s)+r\xx)$ for all integer $1\le r \le 2k$
\label{l:overq}
\Comment{$2k$ queries}
\EndParallel\label{l13}
\State $Q\gets$ the ordered multiset of $2km$ non-decreasing costs from the above queries
\State $C_{min}(q)\gets$ $(q+1)$-th smallest element of $Q$ \label{l:defcqmin}
\State $p_i\gets$ number of times $i\in[m]$ contributes a cost to the $q$ smallest elements of $Q$
\smallskip
\Parallel{ for all links $i\in[m]$}
\If{$f_i(n_i(\s)-\lfloor\frac{n_i(\s)}{\xx}\rfloor\cdot\xx)> C_{min}(q)$}
\label{l:toofewa}
\Comment{$1$ query; only relevant for link $\linka$}
\State $q_i\gets \lfloor\frac{n_i(\s)}{\xx}\rfloor$
\Else{ (using binary search on $q_i\in[0,\min\{km,\lfloor\frac{n_i(\s)}{\xx}\rfloor\}]$)}
\State $q_i \gets  \min \left\{ q_i
 : 
 f_i(n_i(\s) - q_i\xx) \le C_{min}(q) \right\}
$ \label{l:ubqi}
\Comment{$\log(km)$ queries}
\EndIf
\EndParallel 
\smallskip
\If{$\sum_{i\in[m]} q_i = q$} 
\State modify $\s$ by removing $q_i$ and adding $p_i$ blocks of $\xx$ players to every link $i\in [m]$ \State \Return{$\s$}
\ElsIf{$\sum_{i\in[m]} q_i < q$}
\State \Return{\textsc{RefineProfile}($\s,\xx,\qmin,q-1$)}
\Else{ ($\sum_{i\in[m]} q_i > q$)}
\State \Return{\textsc{RefineProfile}($\s,\xx,q+1,\qmax$)}
\EndIf
\EndFunction
\end{algorithmic}
\caption{\textsc{ParallelLinks}}
\label{fig_alg1}
\end{algorithm}

We start with an informal description of algorithm \textsc{ParallelLinks}. On
Line~\ref{l:bestlink} we initialize the algorithm by using one over-query to
find the cheapest link $\linka$, and assigning all $n$ players to link $\linka$.
Note that $\linka$ is the special link, as discussed earlier.
The algorithm then works in $T+1$ phases, where $T=\lfloor\frac{\log(n)}{\log(k)}\rfloor$. Each phase is one iteration of the \forb-loop. 
The \forb-loop is governed by a variable $t$, which is initially $T$ and decreases by $1$ in each iteration.
Within any iteration, the algorithm uses the function \textsc{RefineProfile} to
transform a $k^{t+1}$-equilibrium into a $k^t$-equilibrium. 

Recall, from the overview, that when $k = 2$, we observed that each link can
receive at most one block when we transform a $2^{t+1}$-equilibrium into a
$2^t$-equilibrium. In the following lemma, we establish a similar property for
the case where $k \ne 2$: each link can receive at most $2k$ blocks.
Intuitively, one might expect each link to receive at most $k$ blocks, but the
extra factor of two here arises due to the special link $\linka$, which was not
considered in our simplified overview.
\begin{lemma}
\label{l:numberofmoves}
We can convert a $k^{t+1}$-equilibrium $\s$ into a $k^t$-equilibrium $\s'$
by moving at most $2k$ blocks of $\xx = k^t$ players to any individual link
and at most $km$ blocks of $\xx$ players in total.
\end{lemma}
\begin{proof}
Since $\s$ is $k^{t+1}$-equilibrium, we have
$f_i(n_i(\s)) \le  f_j(n_j(\s)+k^{t+1})$ for all $i \in[m]\setminus\{a\},j\in[m]$. Moreover, either (a) $f_a(n_a(\s)) \le  f_j(n_j(\s)+k^{t+1})$ for all $j\in[m]$ or (b) $n_a(\s)<k^{t+1}$.
In case (a), this implies that each link $j\in[m]$ can in total receive at most
$k$ blocks of size $\xx=k^t$ from links $i\in[m]$.
In case (b), this implies that each link $j\in[m]$ can in total receive at most
$k$ blocks of size $\xx=k^t$ from links $i\in[m]\setminus\{a\}$. Moreover, since
$n_a(\s')<k^{t+1}$, we can move at most $k$ blocks of size $\xx=k^t$ from link $a$.
In either case, in total we move at most $km$ blocks.  
All links receive and lose players only in multiples of~$\xx=k^t$, which ensures that $k^t\vert n_i(\s')$ for all $i\in[m]\setminus\{\linka\}$ is maintained.
\end{proof}


\textsc{RefineProfile} determines the number of blocks $q$ which have to be
moved  by binary search on $q$ in $[0,km]$.  Since, by Lemma
\ref{l:numberofmoves}, each link receives at most $2k$ blocks of players, we
spend $2k$ over-queries to determine the cost function values
$f_i(n_i(\s)+r\cdot \xx)$ for all integers~$r\le 2k$ and all links $i\in[m]$.
We define $Q$ as the multi-set of these cost function values and $C_{min}(q)$
as the $(q+1)$-th smallest value in $Q$. Intuitively,  $C_{min}(q)$ is the cost of
the $(q+1)$-th block of players that we would move. We use $C_{min}(q)$ to find
out how many blocks of players $q_i$ we need to remove from each link $i\in[m]$
so that on each link $i\in[m]$ the cost is at most $C_{min}(q)$ or we can't
remove any further blocks as there are less than $\xx$ players assigned to it
(which can only happen on link $\linka$). By Lemma \ref{l:numberofmoves}, we
need to remove at most $km$ blocks of players in total. Therefore, we can
determine $q_i\in[0,\min\{km,\lfloor\frac{n_i(\s)}{\xx}\rfloor\}]$ by binary
search in parallel on all links, with $\bigO(\log(km))$ under-queries. Now, if
$\sum_{i=1}^m q_i = q$, we can construct a $k^t$-equilibrium by removing $q_i$
and adding $p_i$ blocks of $\xx$ players to link $i\in[m]$; note that for every
$i\in[m]$, either $q_i=0$ or $p_i=0$. If $\sum_{i=1}^m q_i \ne q$, our guess
for $q$ was not correct and we have to continue the binary search on $q$.

The algorithm maintains the following invariant:
\begin{lemma}
\label{l:postcondition}
\textsc{RefineProfile}$(\s,\xx,0,km)$ returns a $\xx$-equilibrium.
\end{lemma}
\begin{proof}
Observe that $\delta=k^t$. 
In the first iteration of the \forb-loop $t=T$ and  \textsc{RefineProfile}$(\s,\xx,0,km)$ gets a $n$-equilibrium as input, which is also a $k^{T+1}$-equilibrium as all players are assigned to link 
$\linka$ and $k^{T+1}>n$. So to prove the claim, it suffices to show that \textsc{RefineProfile}$(\s,k^t,0,km)$ returns a $k^t$-equilibrium if $\s$ is a $k^{t+1}$-equilibrium.
For the $\s$ returned by \textsc{RefineProfile} and the $q$ in its returning
call, we have $f_i(n_i(\s)) \le C_{min}(q) \le  f_i(n_i(\s) + \xx)$ for all
$i\in[m]\setminus\{\linka\}$.
The left inequality follows from line~\ref{l:ubqi} of the algorithm.
The right
inequality follows from the definition of $C_{min}(q)$ as the $(q+1)$-th
smallest element in $Q$ in line~\ref{l:defcqmin} of the algorithm.
For link~$\linka$, we have 
$f_\linka(n_\linka(\s)) \le C_{min}(q) \le  f_\linka(n_\linka(\s) + \xx)$ 
or we have
$f_\linka(n_\linka(\s)) > C_{min}(q)$ and $n_\linka(\s) < \xx$, where the
first case follows from lines~\ref{l:ubqi} and \ref{l:defcqmin} as before, and the second case
corresponds to line~\ref{l:toofewa}.
Noting that \textsc{RefineProfile} maintains that for the returned
$s$ we have $\xx \vert n_i(\s)$ for all $i\in[m]\setminus\{\linka\}$, as it
only moves blocks of size $\xx$, the claim follows.
\end{proof}
%
%



We now give the payoff query complexity of \textsc{RefineProfile}. We split our
analysis into over-queries and non-over-queries (ie.\ under-queries or
normal-queries), because we will later show how the over-queries made by our
algorithm can be translated into a sequence of non-over-queries.

\begin{lemma}
\label{l:rp_queries}
\textsc{RefineProfile($\s,\xx,0,km$)} can be implemented to make $2k$
over-queries and $\bigO(\log^2(km))$ non-over-queries. 
\end{lemma}
\begin{proof}
Note that, as long as $\delta$ is not changed, the queries made on
line~\ref{l:overq} are the same for each pair of $\qmin$ and $\qmax$. Therefore,
we can perform these $2k$ over-queries when
we first call \textsc{RefineProfile($\s,\xx,0,km$)}, and reuse these values during each
recursive call. For each value of $q$ in the binary search, we make
$\bigO(\log(km))$ under-queries to determine the $q_i$'s in parallel for all
links $i\in[m]$.  The binary search on $q$ adds a factor $\log(km)$ to give
$\bigO(\log^2(km))$ under-queries in total.
\end{proof}

Using Lemmas~\ref{l:postcondition} and~\ref{l:rp_queries} we can prove the
following.

\begin{theorem}\label{thm:UBparallel}
Algorithm \textsc{ParallelLinks} returns a pure Nash equilibrium and can be
implemented with $\bigO\left(\log(n)\cdot \frac{\log^2(m)}{\log\log(m)}\right)$
queries, of which $2k \cdot \frac{\log n}{\log \log m}$ are over-queries.
\end{theorem}
\begin{proof}
In the last iteration of the \forb-loop, we have $\xx=1$,
so Lemma~\ref{l:postcondition} implies that~$\s$ is a pure Nash equilibrium.
To find the best link in line \ref{l:bestlink} of the algorithm, we need one
over-query.
For any $k\ge 2$, the algorithm does
$T+1=\bigO\left(\frac{\log(n)}{\log(k)}\right)$ iterations of the \forb-loop. In
each iteration we do $\bigO(\log^2(km))$ under-queries and $2k$ over-queries.
Choosing $k=\Theta(\log(m))$ yields the stated upper bound.
\end{proof}

\paragraph{\bf Using only normal-queries.} We now show how
Algorithm~\ref{fig_alg1} can be implemented without the use of over-queries.
Before doing so, we remark that in the parallel links setting, we can also
avoid using under-queries.

\begin{lemma}
If a parallel links congestion game has at least two links, then every
under-query can be translated into two normal-queries.
\end{lemma}
\begin{proof}
Suppose that the game has $m \ge 2$ links, and let $q = (i_1, i_2, \dots, i_m)$
be an under-query. Let $n' = \sum_{j=1}^m i_j$ be the total number of
players used by $q$. We define the following queries:
\begin{align*}
q_1 &= (i_1 + n - n', i_2, \dots, i_m), \\
q_2 &= (i_1, i_2, \dots, i_m + n - n'). 
\end{align*}
Clearly both $q_1$ and $q_2$ are normal-queries. Query $q_1$ tells us the cost
of links $2$ through $m$ under $q$, and query $q_2$ tells us the cost of link
$1$ under $q$. 
\end{proof}

We now turn our attention to over-queries. The following lemma gives a general
method for translating over-queries into non-over-queries.

\begin{lemma}
\label{lem:overq}
Suppose we have a parallel links game with $m$ links and $n$ players. Let $q =
(i_1, i_2, \dots, i_m)$ be an over-query, and define $n' = \sum_{j=1}^m i_j$. We
can translate $q$ into a sequence of $\bigO(n'/n)$ non-over-queries.
\end{lemma}
\begin{proof}
Consider the following greedy algorithm: find the smallest index $b$ such that
$\sum_{1 \le k \le b} i_k \le n$ and assign links $1$ through $b$ to
query $q_1$. Set $i_1 = i_2 = \dots = i_b = 0$, and repeat. Clearly each query
that we generate during this algorithm is a non-over-query.

Let $q_1, q_2, \dots, q_l$ be the sequence of non-over-queries generated by the
above algorithm for some $l\in\mathbb{N}$. For each $j$, let $n_j$ be the total
number of players used by $q_j$, and observe that $\sum_{1 \le j \le l} n_j = n'$. 
Furthermore, for each
$j$, let $r_j = n - n_j$ be the total number of players \emph{not} used by
$q_j$. Due to the nature of our algorithm, for every $j > 1$ we must have
$r_{j-1} <
n_{j}$, since the first link assigned to $q_j$ would not fit in $q_{j-1}$.
Thus, we have:
\begin{align*}
\sum_{1 \le j \le l} r_{j} &< \sum_{2 \le j \le l} n_j + r_l \\
&< n' + n.
\end{align*}
Since the total number of queries in the sequence is $l$, we can argue that:
\begin{align*}
l &= \frac{1}{n} \sum_{1 \le j \le l} (n_j + r_j) \\
&< \frac{n' + n' + n}{n} \\
&= 1 + \frac{2n'}{n}.
\end{align*}
Thus, our greedy algorithm generates at most $\bigO(n'/n)$ non-over-queries. 
\end{proof}

In order to optimise the number of non-over queries we have to adjust Algorithm \ref{fig_alg1} 
slightly, because with $k=\Theta(\log(m))$ in early iterations of the for loop, i.e., when $T$ is large, the number of players
used in the over queries in line (\ref{l:overq}) is large and applying Lemma
\ref{lem:overq} would 
yield to a total of $\bigO\left(\log(n)\cdot \frac{\log^2(m)}{\log\log(m)}+m\log(m)\right)$ non-over 
queries. 
In contrast, we will now show that our adjusted Algorithm \ref{fig_alg1b} can be implemented 
to do at most $\bigO\left(\log(n)\cdot \frac{\log^2(m)}{\log\log(m)}+m\right)$ non-over 
queries. The main idea is to divide the block size by $2$ until the number of
players in a block is small enough and then switch to $k=\Theta(\log(m))$.

\begin{algorithm}
\algblock[Parallel]{Parallel}{EndParallel}
\algblockdefx[PARALLEL]{PARALLEL}{ENDPARALLEL}%
    [1][Unknown]{Parallel~{}#1}%
    {EndParallel}
\begin{algorithmic}[1]
\State $a\gets \arg\min_{i\in[m]} f_i(n)$
\Comment{1 over-query}
\label{l:bestlinkb}
\State initialize strategy profile $\s$ by putting all players on link $a$ 
\State $T \gets \left\lfloor\frac{\log(n/m)}{\log(k)}\right\rfloor$
\State $T_0\gets$ largest $t$ such that $k^T 2^t<n$
\For{$t= T_0, T_0-1, \ldots, 1$}
\State $\delta \gets k^T 2^t$  
\State $\s \gets$ \textsc{RefineProfile($\s,\xx,0,2m$)}  
\EndFor
\For{$t= T, T-1, \ldots, 1, 0$}
\State $\delta \gets k^t$  
\State $\s \gets$ \textsc{RefineProfile($\s,\xx,0,km$)}  
\EndFor
\State
\Return \s
\end{algorithmic}
\caption{\textsc{ParallelLinks avoiding over-queries}}
\label{fig_alg1b}
\end{algorithm}

To initialize our algorithm, we make an over-query that uses $m \cdot n$
players. By Lemma~\ref{lem:overq}, we can translate this into $\bigO(m)$ 
non-over-queries. 

In each iteration of the first \forb-loop with value $t$, by Lemma \ref{l:rp_queries}, we make $\bigO(1)$ over-queries. 
Each of these uses at most $n+m\cdot 4 \cdot k^T2^t$ players. By Lemma~\ref{lem:overq}, these can be simulated by 
$\bigO(1+\frac{m  k^T2^t}{n})$ non-over-queries. Summing up over all iterations and using the definition of $T_0$, we can argue that  all over-queries of the first \forb-loop can be 
simulated by 
\begin{align*}
\sum_{t=1}^{T_0} \bigO\left(1+\frac{m  k^T2^t}{n}\right) &= \bigO(T_0)+\bigO\left(\frac{m  k^T2^{T_0}}{n} \right)  = \bigO(m)
\end{align*}
non-over-queries.

In each iteration of the second \forb-loop with value $t$, by Lemma \ref{l:rp_queries}, we make make $2k$
over-queries that each use at most $n + m \cdot 2k \cdot k^t$ players. 
By Lemma~\ref{lem:overq}, these can be simulated by 
$\bigO(\frac{mk^{t+1}}{n})$ non-over-queries. Summing up over all iterations, we can argue that all over-queries of the second \forb-loop can be 
simulated by 
\begin{align*}
\sum_{t=0}^{\lfloor\frac{\log(n/m)}{\log(k)}\rfloor} \bigO\left(\frac{m k^{t+1}}{n}\right) 
&= \bigO\left(\frac{m}{n} \cdot  k^{\frac{\log(n/m)}{\log(k)}+1}\right)  \\
&= \bigO\left(\frac{m}{n} \cdot  k^{\frac{\log(n)-\log(m)+\log(k)}{\log(k)}}\right)  \\
&= \bigO\left(\frac{m}{n} \cdot  k^{\frac{\log(n)}{\log(k)}}\right)  \\
&= \bigO(m)
\end{align*}
non-over-queries.

Combining this discussion with Theorem~\ref{thm:UBparallel}, we get the following result:
\begin{theorem}
\label{thm:UBparallel_n}
Algorithm \ref{fig_alg1b} returns a pure Nash equilibrium and can be
implemented with $\bigO\left(\log(n)\cdot \frac{\log^2(m)}{\log\log(m)}+m\right)$
queries.
\end{theorem}

The upper bound in Theorem~\ref{thm:UBparallel_n} should be contrasted with the
lower bound of $\log(n) + m$ (Corollary~\ref{cor:LBlogn+m}).


\subsection{Symmetric Network Congestion Games on Directed Acyclic Graphs}
\label{sec:dags}


In this section, we consider symmetric network congestion games 
on directed acyclic graphs. Throughout this section, we consider the game
$\Gamma=(\N, V, E, (f_e)_{e\in E}, o, d)$, where $(V, E)$ is a directed acyclic
graph (DAG). We use
the $\prec$ relation to denote a topological ordering over the vertices in~$V$.
We assume that, for every vertex $v \in V$, there exists a path
from $o$ to~$v$, and there exists a path from $v$ to $d$. If either of these
conditions does not hold for some vertex~$v$, then $v$ cannot appear on an
$o$-$d$ path, and so it is safe to delete~$v$.

We provide an algorithm that discovers a cost function for each edge. One
immediate observation is that we can never hope to find the actual cost
functions. Consider the following one-player congestion game.
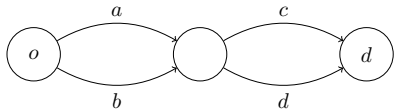
\begin{center}
\resizebox{0.35\textwidth}{!}{
\begin{tikzpicture}[node distance=3cm]
\node[state] (o) {$o$};
\node[state] [right of=o] (m) {};
\node[state] [right of=m] (d) {$d$};
\path[->]
	(o) [bend right] edge node[auto,swap] {$b$} (m)
	(o) [bend left] edge node[auto] {$a$} (m)
	(m) [bend right] edge node[auto,swap] {$d$} (d)
	(m) [bend left] edge node[auto] {$c$} (d)
	;
\end{tikzpicture}
} 
\end{center}
If we set $f_a(1) = f_b(1) = 1$ and $f_c(1) = f_d(1) = 0$, then all $o$-$d$
paths have cost~$1$. However, we could also achieve the same property by
setting $f_a(1) = f_b(1) = 0$ and setting $f_c(1) = f_d(1) = 1$. Thus, it is
impossible to learn the actual cost functions using payoff queries.

To deal with this issue, we introduce the notion of an \emph{equivalent} cost
function: two cost functions are said to be equivalent if they assign the same
cost to every strategy profile. We show that, while it is impossible to find
the actual cost function via payoff queries, we can use payoff queries to find
an equivalent cost function.

Our algorithm proceeds inductively over the number of players in the game.
For the base case, we give an algorithm that finds an equivalent cost function
$f'$ such that $f'_e(1)$ is defined for every edge $e$. This corresponds to
learning all the costs in a one-player congestion game played on $\Gamma$. Then,
for the inductive step, we show how the costs for an $i$-player game can be used
to find the costs in an $i+1$ player game. That is, we use the known values
of $f'_e(j)$ for $j \le i$ to find the cost of $f'_e(i+1)$ for every edge $e$.
Therefore, at the end of the algorithm, we have an equivalent cost function
$f'$ for an $n$-player game on $\Gamma$, and we can then apply a standard
congestion game algorithm \citep{FPT04} in order to solve our game.

Unlike our work on parallel links, in this section we will not use over-queries
at all.
In each inductive step, when we are considering an $i$-player congestion game,
we will make queries that use exactly $i$ players. Thus, in the first $n-1$
rounds we will use under-queries, and in the final round we will use
normal-queries. For the sake of brevity, in this section we will use the word
``query'' to refer to both normal and under-queries.

As a shorthand for defining queries, we use notation of the form $\s \leftarrow
(1 \mapsto p, 3 \mapsto q)$. This example defines $\s$ to be a four-player query
that assigns 1 player to $p$ and 3 players to $q$, where $p$ and $q$ are paths
from the origin to the destination in a symmetric network congestion game. We
use $\query(\s)$ to denote the outcome of querying $\s$. It returns a
function~$c_\s$, which gives the cost of each strategy when $\s$ is played.

\paragraph{Preprocessing.} Our algorithm requires a preprocessing step. We say
that edges $e$ and~$e'$ are \emph{dependent} if visiting one implies that we
must visit the other. More formally, $e$ and $e'$ are dependent if, for every
$o$-$d$ path $p$, we either have $e, e' \in p$, or we have $e, e' \notin p$. We
preprocess the game to ensure that there are no pairs of dependent edges. To do
this, we check every pair of edges $e$ and $e'$, and test whether they are
dependent. If they are, then we \emph{contract} $e'$, i.e., if $e' = (v, u)$,
then we delete $e'$, and set $v = u$. The following lemma shows that this
preprocessing is valid, and therefore, from now on, we can assume that our
congestion game contains no pair of dependent edges.

\begin{lemma}
\label{LEM:PREPROCESS}
There is an algorithm that, given a congestion game $\Gamma$, where $(V, E)$ is
a DAG, produces a game $\Gamma'$ with no pair of dependent edges, such that
every Nash equilibrium of~$\Gamma'$ can be converted 
to a Nash equilibrium of
$\Gamma$. The algorithm and conversion 
of equilibria 
take polynomial time and make zero payoff queries.
Moreover, payoff queries to $\Gamma'$ can be trivially simulated with payoff
queries to $\Gamma$.
\end{lemma}
\begin{proof}
Our algorithm will check, for each pair of edges $e = (v, u)$ and $e' = (v',
u')$, whether~$e$ and $e'$ are dependent. This is done in the following way.
Note that if $v = v'$, then $e$ and $e'$ cannot possibly be dependent. Thus, we
can assume without loss of generality that $v \prec v'$. The algorithm performs
two checks:
\begin{itemize}
\item Delete $e$ and verify that there is no path from $o$ to $v'$.
\item Delete $e'$ and verify that there is no path from $u$ to $d$.
\end{itemize}
The first check ensures that every path that uses~$e'$ must also use $e$. The
second check ensures that every path that uses~$e$ must also use~$e'$. Thus, if
both checks are satisfied, then~$e$ and~$e'$ are dependent. On the other hand,
if one of the checks is not satisfied, then we can construct an $o$-$d$ path
that uses~$e$ and not~$e'$, or a path that uses~$e'$ and not~$e$, which verifies
that~$e$ and~$e'$ are not dependent. 

Whenever the algorithm finds a pair of edges $e, e' \in E$ that are dependent,
it contracts~$e'$. More formally, if $e' = (v, u)$, then the algorithm
constructs a new congestion game $\Gamma' =(\N, V', E', (f'_e)_{e\in E'}, o, d)$
 where $V' = V \setminus \{u\}$, and $E'$ contains:
\begin{itemize}
\item every edge $(w, x) \in E$ with and $w \ne u$, and
\item an edge $(v, x)$ for every edge $(u, x) \in E$.
\end{itemize}
Note that $E'$ does not contain $e'$.
Moreover, we define the cost functions $f'$ as follows. For each edge $e'' \ne
e$, we set $f'_{e''}(i) = f_{e''}(i)$ for all $i$. For the edge $e$, we define
$f'_e(i) = f_e(i) + f_{e'}(i)$ for all $i$.

We argue that this operation is correct. Since $e$ and $e'$ are dependent, we
have that, for every strategy profile $\s$, and for every $o$-$d$ path $p$:
\begin{equation*}
\sum_{e'' \in p} f'_{e''}(i) = \sum_{e'' \in p} f_{e''}(i).
\end{equation*}
Therefore, we can easily translate every Nash equilibrium of $\Gamma'$ into a
Nash equilibrium for~$\Gamma$.  Moreover, every payoff query for $\Gamma'$ can
be translated into a payoff query for $\Gamma$ by adding the edge $e'$ where
appropriate.

Thus, the algorithm constructs a sequence of games $\Gamma_1$, $\Gamma_2$,
$\dots$, where each game $\Gamma_{i+1}$ is obtained by contracting an edge in
$\Gamma_{i}$. Moreover, the Nash equilibria for $\Gamma_{i+1}$ can be translated
to $\Gamma_{i}$, which implies that the algorithm is correct. This algorithm can
obviously be implemented in polynomial time. Moreover, since the algorithm only
inspects structural properties of the graph, it does not make any payoff
queries. 
\end{proof}

\paragraph{Equivalent cost functions.}


As we have mentioned, we cannot hope to find the actual cost function of
$\Gamma$ using payoff queries. To deal with this, we introduce the following
notion of equivalence.
\begin{definition}[Equivalence]
Two cost functions $f$ and $f'$ are equivalent if for every strategy profile $\s
= (s_1, s_2, \dots, s_n)$, we have $\sum_{e \in s_i} f_e(n_e(\s)) = \sum_{e
\in s_i} f'_e(n_i(\s))$, for all $i$.
\end{definition}
Clearly, the Nash equilibria of a game cannot change if we replace its cost
function~$f$ with an equivalent cost function~$f'$.

We say that $(f'_e)_{e\in E}$ is a \emph{partial} cost function if for some
$e \in E$ and some $i \le n$, $f'_e(i)$ is undefined. We say that $f''$ is an
\emph{extension} of $f'$ if $f''$ is a partial cost function, and if $f''_e(i) =
f'_e(i)$ for every $e \in E$ and $i \le n$ for which $f'_e(i)$ is defined. We
say that $f''$ is a \emph{total extension} of $f'$ if $f''$ is an extension of
$f'$, and if $f''_e(i)$ is defined for all $e \in E$ and all $i \le n$.
\begin{definition}[Partial equivalent cost function]
Let~$f$ be a cost function. We say that~$f'$ is a partial equivalent of~$f$
if~$f'$ is a partial cost function, and if there exists a total extension~$f''$
of~$f'$ such that $f''$ is equivalent to~$f$.
\end{definition}

Our goal is to find a total equivalent cost function by learning the costs one
edge at a time. Thus, our algorithm will begin with a partial cost function
$f^0$ such that $f^0_e(i)$ is undefined for all $e \in E$ and all $i \le n$.
Since it is undefined everywhere, it is obvious that $f^0$ is a partial
equivalent of $f$. At every step of the algorithm, we will take a partial
equivalent cost function $f'$ of $f$, and produce an extension $f''$ of $f'$,
such that $f''$ is still a partial equivalent of $f$. This guarantees that, when
the algorithm terminates, the final cost function is equivalent to $f$.

\subsection{The One-Player Case}\label{sec:dags2}

\paragraph{Outline.}

For the one player case, our algorithm is relatively straightforward. The
algorithm proceeds iteratively by processing the vertices according to their
topological order, starting from the origin vertex $o$, and moving towards the
destination vertex $d$. Each time we process a vertex $k$, we determine the cost
of every incoming edge $(u, k)$. There are two different cases: the case where
$k \ne d$, and the case where $k = d$. For the latter case, we will observe
that, once we know the cost of every edge other than the incoming edges to $d$,
we can easily find the cost of the incoming edges to $d$. 

The former case is slightly more complicated. When we consider a vertex
$k \ne d$, it turns out that we cannot find the actual costs for the incoming
edges at $k$. Instead, we can use payoff queries to discover the difference in
cost between each pair of incoming edges, and therefore, we can find the
cheapest incoming edge $e$ to $k$. We proceed by fixing the cost of $e$ to be
$0$. Once we have done this, we can then set the cost of each other incoming
edge $e'$ according to the difference between the cost of $e$ and the cost of
$e'$, which we have already discovered. We prove that this approach is correct
by showing that it yields a partial equivalent cost function.

We now formally describe our algorithm. The algorithm begins with the partial
cost function $f^0$. The algorithm processes vertices iteratively according to
the topological ordering $\prec$. Suppose that we are in iteration $a+1$ of the
algorithm, and that we are processing a vertex $k \in V$. We have a partial
equivalent cost function $f^{a}$ such that $f^a_e(1)$ is defined for every edge
$e = (v, u)$ with $u \prec k$, for some vertex~$k$. We then produce a partial
equivalent cost function $f^{a+1}$ such that $f^{a+1}_e(1)$ is defined for every
edge $e = (v, u)$ with $u \preceq k$. We now consider the two cases.

\paragraph{The $k \ne d$ case.}

\renewcommand{\algorithmicrequire}{\textbf{Input:}}
\renewcommand{\algorithmicensure}{\textbf{Output:}}

\begin{algorithm}
\begin{algorithmic}[1]
\Require A partial equivalent cost function $f^a$, such that $f^{a}_e(1)$ is defined for all edges $(v, u)$ with $u \prec k$.
\Ensure A partial equivalent cost function $f^{a+1}$, such that $f^{a+1}_e(1)$ is defined for all edges $(v, u)$ with $u \preceq k$.

\ForAll {$e$ for which $f^{a}_e(1)$ is defined}
\label{l:copystart}
\State $f^{a+1}_e(1) \leftarrow f^{a}_e(1)$
\EndFor
\label{l:copyend}

\State $p \leftarrow $ an arbitrary $k$-$d$ path

\ForAll {$e = (v, k) \in E$}
\label{l:tailstart}
\State $p' \leftarrow$ an arbitrary $o$-$v$ path
\State $\s \leftarrow (1 \mapsto p'ep)$
\State $c_\s \leftarrow \query(\s)$
\State $t(ep) \leftarrow c_{\s}(p'ep) - \sum_{e' \in p'} f^{a}_{e'}(1)$
\label{l:tep}
\EndFor
\label{l:tailend}
\State $e' \leftarrow $ edge $e = (v, k)$ that minimises $t(ep)$
\label{l:linem}
\State $f^{a+1}_{e'}(1) \leftarrow 0$
\label{l:zero}
\ForAll {$e = (v, k) \in E$ with $e \ne e'$}
\label{l:relstart}
\State $f^{a+1}_e(1) \leftarrow t(ep) - t(e'p)$
\label{l:diff}
\EndFor
\label{l:relend}
\end{algorithmic}
\caption{\textsc{ProcessK}}
\label{alg:k}
\end{algorithm}

We use the procedure shown in Algorithm~\ref{alg:k} to process $k$.
Lines~\ref{l:copystart} through~\ref{l:copyend} simply copy the old cost
function $f^a$ into the new cost function $f^{a+1}$. This ensures that $f^{a+1}$
is an extension of $f^a$. The algorithm then picks an arbitrary $k$-$d$
path~$p$. The loop on lines~\ref{l:tailstart} through~\ref{l:tailend} compute
the function $t$, which for each incoming edge $e = (v, k)$, gives the cost
$t(ep)$ of allocating one player to $ep$. Note, in particular, that the value of
the expression $\sum_{e' \in p'} f^a_{e'}(1)$ is known to the algorithm, because
every vertex visited by $p'$ has already been processed. The algorithm then
selects~$e'$ to be the edge that minimises~$t$, and sets the cost of $e'$ to be
$0$. Once it has done this, lines~\ref{l:relstart} through~\ref{l:relend}
compute the costs of the other edges relative to $e'$. 

When we set the cost of~$e'$ to be~$0$, we are making use of equivalence.
Suppose that the actual cost of~$e'$ is $c_{e'}$. Setting the cost of~$e'$ to
be~$0$ has the following effects:
\begin{itemize}
\item Every incoming edge at~$k$ has its cost reduced by~$c_{e'}$. 
\item Every outgoing edge at~$k$ has its cost increased by~$c_{e'}$.
\end{itemize}
This maintains equivalence with the original cost function, because for every
path $p$ that passes through $k$, the total cost of $p$ remains unchanged. The
following lemma formalises this and proves that $f^{a+1}$
is indeed a partial equivalent cost function.

\begin{lemma}
\label{LEM:ALGK}
Let $k \ne d$ be a vertex, and let $f^a$ be a partial equivalent cost function
such that $f^a_e(1)$ is defined for all edges $e = (v, u)$ with $u \prec k$.
When given these inputs, Algorithm~\ref{alg:k} computes a partial equivalent
cost function $f^{a+1}$ such that $f^{a+1}_e(1)$ is defined for all edges $e =
(v, u)$ with $u \preceq k$.
\end{lemma}
\begin{proof}
It can be verified that the algorithm assigns a cost to $f^{a+1}_e(1)$ for every
edge $e = (v, u)$ with $u \preceq k$. To complete the proof of the lemma, we
must show that~$f^{a+1}$ is a partial equivalent cost function. Since~$f^{a}$ is
a partial equivalent cost function, there must exist a total extension
of~$f^{a}$ that is equivalent to~$f$. Let~$f'$ denote such an extension. We 
use~$f'$ to construct~$f''$, which is a total extension of~$f^{a+1}$ that is
equivalent to~$f$.

Let $e = (v, k)$ be an incoming edge at $k$. We begin by deriving a formula for
$t(ep)$, which is computed on line~\ref{l:tep}. Note that, since $f'$ is
equivalent to $f$, we have $c_{\s}(p'ep) = \sum_{e' \in p' e p} f'_{e'}(1)$.
Note also that $f'_{e'}(1) = f^{a}_{e'}(1)$ for every edge $e' \in p'$.
Therefore, we have the following:
\begin{align*}
t(ep) &= c_{\s}(p'ep) - \sum_{e' \in p'} f^{a}_{e'}(1) \\
&=  \sum_{e' \in p' e p} f'_{e'}(1) - \sum_{e' \in p'} f'_{e'}(1) \\
&= \sum_{e' \in e p} f'_{e'}(1).
\end{align*}

For each edge $e = (v, k)$ with $e \ne e'$, line~\ref{l:diff} sets:
\begin{align*}
f^{a+1}_e(1) &= t(ep) - t(e'p) \\
&= \sum_{e' \in e p} f'_{e'}(1) -  \sum_{e' \in e' p} f'_{e'}(1) \\
&= f'_{e}(1) - f'_{e'}(1).
\end{align*}
Note also that line~\ref{l:zero} sets:
\begin{equation*}
f^{a+1}_{e'}(1) = 0 = f'_{e'}(1) - f'_{e'}(1).
\end{equation*}
Hence, we can conclude that $f^{a+1}_e(1) = f'_{e}(1) - f'_{e'}(1)$ for every
incoming edge $e = (v, k)$.

We construct the total cost function $f''$ as follows. For every edge $e = (v,
u)$, and every $i \le n$, we set:
\begin{equation*}
f''_e(i) = \begin{cases}
f'_e(i) - f'_{e'}(1) & \text{if $u = k$,} \\
f'_e(i) + f'_{e'}(1) & \text{if $v = k$,} \\
f'_e(i) & \text{otherwise.} \\
\end{cases}
\end{equation*}
Since we have shown  that $f^{a+1}_e(1) = f'_{e}(1) - f'_{e'}(1)$ for every
incoming edge $e = (v, k)$, we have that $f''_e(1)$ is a total extension
of $f^{a+1}$.

We must now show that $f''_e$ and $f$ are equivalent. We will do this by showing
that $f''$ and $f'$ are equivalent. Let $\s = (s_1, s_2, \dots, s_n)$ be an
arbitrarily chosen strategy profile. If $s_i$ does not visit $k$, then we
have:
\begin{equation*}
\sum_{e \in s_i} f''_e(n_e(\s)) = \sum_{e \in s_i} f'_e(n_e(\s)).
\end{equation*}
On the other hand, if $s_i$ does visit $k$, then it must use exactly one edge
$(v, u)$ with $u = k$, and exactly one edge $(v, u)$ with $v = k$. Therefore, we
have:
\begin{align*}
\sum_{e \in s_i} f''_e(n_e(\s)) &= \sum_{e \in s_i} f'_e(n_e(\s)) - f'_{e'}(1) +
f'_{e'}(1) \\
&= \sum_{e \in s_i} f'_e(n_e(\s)).
\end{align*}
Therefore, $f''$ is equivalent to $f'$, which also implies that it is equivalent
to $f$. Thus, we have found a total extension of $f^{i+1}$ that is equivalent to $f$, as required.
\end{proof}

\paragraph{The $k = d$ case.}

When the algorithm processes $d$, it will have a partial cost function $f^{a}$
such that $f^{a}_e(1)$ is defined for every edge $e = (v, u)$ with $u \ne d$.
The algorithm is required to produce a partial cost function $f^{a+1}$ such that
$f^{a+1}_e(1)$ is defined for all $e \in E$. We use Algorithm~\ref{alg:d} to
do this.
\begin{algorithm}
\begin{algorithmic}[1]
\Require A partial equivalent cost function $f^{a}$, such that $f^{a}_e(1)$ is
defined for all edges $e = (v, u)$ with $u \prec d$.
\Ensure A partial equivalent cost function $f^{a+1}$, such that $f^{a}_e(1)$ is
defined for all edges $e \in E$.

\ForAll {$e$ for which $f^{a}_e(1)$ is defined}
\label{l:copystart2}
\State $f^{a+1}_e(1) \leftarrow f^{a}_e(1)$
\EndFor
\label{l:copyend2}

\ForAll {$e = (v, d) \in E$}
\State $p \leftarrow$ an arbitrary $o$-$v$ path
\State $\s \leftarrow (1 \mapsto pe)$
\label{l:profile}
\State $c_\s \leftarrow \query(\s)$
\label{l:query}
\State $f^{a+1}_e(1) \leftarrow c_\s(pe) - \sum_{e' \in p} f^{a}_{e'}(1)$
\label{l:costs}
\EndFor
\end{algorithmic}
\caption{\textsc{ProcessD}}
\label{alg:d}
\end{algorithm}
Lines~\ref{l:copystart2} through~\ref{l:copyend2} ensure that $f^{a+1}$ is
equivalent to $f^a$. Then, the algorithm loops through each incoming edge $e =
(v, d)$, and line~\ref{l:costs} computes $f^{a+1}_e(1)$. Note, in particular,
that $f^a_{e'}(1)$ is defined for every edge $e' \in p$, and thus the
computation on line~\ref{l:costs} can be performed. Lemma~\ref{LEM:ALGD} shows
that Algorithm~\ref{alg:d} is correct.

\begin{lemma}
\label{LEM:ALGD}
Let $k \ne d$ be a vertex, and let $f^a$ be a partial equivalent cost function
defined for all edges $(v, u)$ with $u \prec d$. When given these inputs,
Algorithm~\ref{alg:d} computes a partial equivalent cost function $f^{a+1}$. 
\end{lemma}
\begin{proof}
Since $f^{a}$ is a partial equivalent cost function, there must exist a cost
function $f'$ that is an extension of $f^{a}$, where $f'$ is equivalent to $f$.
We show that $f'$ is also an extension of $f^{a+1}$.

Let $e = (v, d)$ be an incoming edge at $d$. Consider line~\ref{l:costs} of the
algorithm. Note that, since $f'$ is equivalent to $f$, we have $c_\s(pe) =  
\sum_{e' \in pe} f'_{e'}(1)$. Furthermore, since $f'$ is an extension of
$f^{a+1}$, we have $f^{a}_{e'}(1) = f'_{e'}(1)$ for every $e' \in p$.
Therefore, we have:
\begin{align*}
f^{a+1}_e(1) &= c_\s(pe) - \sum_{e' \in p} f^{a}_{e'}(1) \\
& = \sum_{e' \in pe} f'_{e'}(1) - \sum_{e' \in p} f'_{e'}(1)\\
& = f'_e(1).
\end{align*}
We also have $f^{a+1}_e(1) = f'_e(1)$ for every edge $e = (v, u)$ with $u
\prec d$, and we have shown that $f^{a+1}_e(1) = f'_e(1)$ for every edge $e =
(v, u)$ with $u = d$. Therefore $f'$ is an extension of $f^{a+1}$, which implies
that $f^{a+1}$ is a partial equivalent cost function.
\end{proof}

\paragraph{Query complexity.} The algorithm makes exactly $|E|$
payoff queries in order to find the one-player costs. 
When Algorithm~\ref{alg:k} processes a vertex $k$, it makes exactly one query
for each incoming edge $(v, k)$ at $k$. The same property holds for
Algorithm~\ref{alg:d}. This implies that, in total, the algorithm makes $|E|$
queries. 

\subsection{The Many-Player Case}\label{sec:dags3}

In this section, we will assume that we have a partial equivalent cost function
$f^a$ such that $f^a_e(j)$ is defined whenever $j \le i$. We will give an
algorithm that goes through a sequence of iterations and produces a partial cost function $f^{a'}$, such that $f^{a'}_e(j)$ is defined whenever $j \le i + 1$.

\paragraph{Outline.} The algorithm for the many-player case proceeds in a
similar fashion to the algorithm for the one-player case. The algorithm is still
iterative, and it still processes vertices according to their topological order,
starting from the origin $o$, and moving towards the destination $d$. In this
algorithm, when we process a vertex $k$, we will discover, for each incoming
edge $e$ to $k$, the cost of placing $i+1$ players on $e$.

However, there is an additional complication. Our technique for discovering the
cost of placing $i+1$ players on the incoming edge at $k$ requires two edge
disjoint paths from $k$ to $d$, but there is no reason at all to assume that two
such paths exist. We say that an edge $e$ is a bridge between two vertices $v$
and $u$, if every $v$-$u$ path contains $e$. Furthermore, if we fix a vertex $k
\in V$, then we say that an edge $e$ is a $k$-bridge if $e$ is a bridge between
$k$-$d$. The following lemma can be proved using the max-flow min-cut theorem
and is a variant of Menger's theorem.
\begin{lemma}
\label{LEM:MENGER}
Let $v$ and $u$ be two vertices. There are two edge disjoint paths between $v$
and $u$ if, and only if, there is no bridge between $v$ and $u$.
\end{lemma}
\begin{proof}
Let $(V, E)$ be a graph, and let $v, u \in V$ be two vertices. We construct a
network flow instance where every edge $e \in E$ has capacity $1$, and we ask
for the maximum flow between $v$ and $u$. Since each edge has capacity $1$, we
have that the maximum flow between $v$ and $u$ is greater than $1$ if, and only
if, there are two edge-disjoint paths between $v$ and~$u$. Moreover, by the
max-flow min-cut theorem, 
the maximum flow from $v$ to
$u$ is greater than $1$ if and only if there is no bridge between $v$ and $u$. 
\end{proof}

As a consequence of Lemma~\ref{LEM:MENGER}, we can only process~$k$
if there are no $k$-bridges. To resolve this, before attempting to
process $k$, we first use a separate algorithm to determine the cost of placing
$i+1$ players on each $k$-bridge. After doing this, we can then find two $k$-$d$
paths that are edge disjoint \emph{except for $k$ bridges}. This, combined with
the fact that we know the cost of placing $i+1$ players on each $k$-bridge, is
sufficient to allow us to process~$k$.

The remainder of this section will proceed as follows. We first describe our
algorithm for finding the costs of the $k$ bridges. After doing so, we then
describe our algorithm for processing $k$.

\paragraph{Bridges.}

Given a vertex $k$, we show how to determine the cost of the $k$-bridges. Let
$b_1$, $b_2$, \dots, $b_m$ denote the list of $k$-bridges sorted according to
the topological ordering $\preceq$. That is, if $b_1 = (v_1, u_1)$, and $b_2 =
(v_2, u_2)$, then we have $v_1 \prec v_2$, and so on. Our algorithm is given a
partial cost function $f^a$, such that $f^{a}_e(j)$ is defined for all $j \le
i$, and returns a cost function $f^{a+1}$ that is an extension of $f^a$ where,
for all $\ell$, we have that $f^{a+1}_{b_\ell}(i+1)$ is defined.

Our algorithm processes the $k$-bridges in reverse topological order, starting with the final bridge~$b_m$.
Suppose that we are processing the bridge~$b_j = (v, u)$. 
We will make one payoff query to find the cost of $b_j$, which is described by
the following diagram.
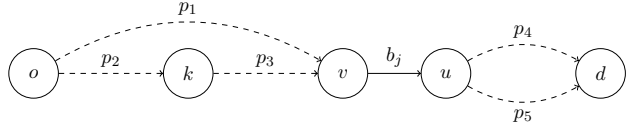
\begin{center}
\resizebox{0.55\textwidth}{!}{
\begin{tikzpicture}[node distance=3cm]
\node [state] (o) {$o$};
\node [state,right of=o] (k) {$k$};
\node [state,right of=k] (a) {$v$};
\node [state,right of=a,node distance=2cm] (b) {$u$};
\node [state,right of=b] (d) {$d$};
\path[->]
	(a) edge node [auto] {$b_j$} (b)
	(b) edge [bend left, dashed] node [auto] {$p_4$} (d)
	(b) edge [bend right, dashed] node [auto,swap] {$p_5$} (d)
	(o) edge [dashed] node [auto] {$p_2$} (k)
	(k) edge [dashed] node [auto] {$p_3$} (a)
	(o) edge [dashed,bend left] node [auto] {$p_1$} (a)
	;
\end{tikzpicture}
}
\end{center}
The dashed lines in the diagram represent paths. They must satisfy some
special requirements, which we now describe. The paths $p_4$ and $p_5$ must be
edge disjoint, apart from $k$-bridges. The following lemma shows that we can
always select two such paths.

\begin{lemma}
\label{LEM:P4P5}
For each $k$-bridge $b_j = (v, u)$, there exists two paths $p_4$ and $p_5$ from $u$ to $d$ such that 
$p_4 \cap p_5 = \{b_{j +1}, b_{j+2}, \dots b_m\}$.
\end{lemma}
\begin{proof}
Note that for each~$\ell$, there cannot exist a bridge between $b_\ell$ and
$b_{\ell+1}$. Therefore, we can apply Lemma~\ref{LEM:MENGER} to argue that there
must exist two edge-disjoint paths between $b_\ell$ and $b_{\ell+1}$ For the
same reason, we can find two edge-disjoint paths between $b_m$ and $d$. To
complete the proof, we simply concatenate these paths. 
\end{proof}

On the other hand, the paths $p_1$, $p_2$, and $p_3$ must satisfy a different
set of constraints, which are formalised by the following lemma.
\begin{lemma}
\label{LEM:P1P3}
Let $b_j = (v, u)$ be a $k$-bridge, let $p_2$ be an arbitrarily chosen $o$-$k$
path. There exists an $o$-$k$ path $p_1$ and a $k$-$v$ path $p_3$ such that:
$p_1$ and $p_3$ are edge disjoint; and 
if $p_1$ visits $k$, then $p_2$ and $p_1$ use different incoming edges for $k$.
\end{lemma}
\begin{proof}
We show how $p_1$ and $p_3$ can be constructed. This splits into two cases, and
we begin by considering the bridges $b_j$ with $j > 1$. Due to our preprocessing
from Lemma~\ref{LEM:PREPROCESS}, $b_j$ and $b_{j-1}$ cannot be dependent. Note
that every $o$-$d$ path that uses~$b_{j-1}$ must also use~$b_j$. Therefore,
there must exist an $o$-$d$ path~$p$ that uses~$b_j$ and not~$b_{j-1}$. We
fix~$p_1$ to be the prefix of~$p$ up to the point where it visits~$b_j$.
Let~$p'_3$ be an arbitrarily selected path from~$k$ to~$b_{j-1}$. Note that
$p_1$ cannot share an edge with~$p'_3$, because otherwise $p_1$ would be forced
to visit $b_{j-1}$. 

We now show how $p'_3$ can be extended to reach $b_j$ without
intersecting~$p_1$. Since there are no bridges between $b_{j-1}$ and $b_{j}$, we
can apply Lemma~\ref{LEM:MENGER} to obtain two edge-disjoint paths $q$ and $q'$
from $b_{j-1}$ to~$b_j$. If one of these paths does not intersect with $p_1$,
then we are done. Otherwise suppose, without loss of generality, that $p_1$
intersects with $q$ before it intersects with $q'$. We create a path $p'_1$ that
follows $p_1$ until the first intersection with $q$, and follows $q$ after that.
Since $q$ and $q'$ are disjoint, the paths $p_1'$ and $p'_3q'$ satisfy the
required conditions.

Now we consider the bridge $b_1$. If $k$ has at least two incoming edges,
then we can apply Lemma~\ref{LEM:MENGER} to find two edge disjoint paths from
$k$ to $b_1$, and we can easily construct $p_1$ and $p_3$ using these paths.
Otherwise, let $e$ be the sole incoming edge at $k$. Since $e$ and $b_1$ are not
dependent, we can find a path~$p_1$ from~$o$ to~$b_1$ which does not use~$e$,
and we can use the same technique as we did for $j > 1$ to find a path $p_3$
from~$k$ to~$b_1$ that does not intersect with~$p_1$. 
\end{proof}

\begin{algorithm}
\begin{algorithmic}[1]
\Require A vertex $k$, and a partial equivalent cost function $f^a$, such that
$f^{a}_e(j)$ is defined for every $j \le i$.
\Ensure A partial equivalent cost function $f^{a+1}$, such that $f^{a+1}$ is an
extension of $f^a$, and $f^{a+1}_e$ is defined for every $e$ that is a $k$
bridge.

\ForAll {$e$ and $j$ for which $f^{a}_e(j)$ is defined}
\label{l:copystart3}
\State $f^{a+1}_e(j) \leftarrow f^{a}_e(j)$
\EndFor
\label{l:copyend3}

\For {$j =$ m to $1$}
\State $p_4$, $p_5 \leftarrow $ paths chosen according to
Lemma~\ref{LEM:P4P5}
\State $p_1$, $p_2$, $p_3$ $\leftarrow$ paths chosen according to
Lemma~\ref{LEM:P1P3}
\State $\s \leftarrow (1 \mapsto p_1 b_j p_4, i \mapsto p_2 p_3 b_j p_5)$
\State $c_\s \leftarrow \query(\s)$
\State $f^{a+1}_{b_j}(i+1) \leftarrow c_\s(p_1 b_j p_4) - \sum_{e \in p_1}
f^{a+1}_e(n_e(\s)) - \sum_{e \in p_4} f^{a+1}_e(n_e(\s))$
\label{l:bridgesum}
\EndFor
\end{algorithmic}
\caption{\textsc{FindKBridges}($k$)}
\label{alg:bridges}
\end{algorithm}

Algorithm~\ref{alg:bridges} shows how the cost of placing $i+1$ players on each
of the $k$-bridges can be discovered. Note that on line~\ref{l:bridgesum}, since
$\s$ assigns one player to $p_1$, we have $n_e(\s) = 1$ for every $e \in p_1$.
Therefore, $f^{a+1}_e(n_e(\s))$ is known for every edge $e \in p_1$.
Moreover, for every edge $e \in p_4$, we have that $n_e(\s) = i + 1$ if $e$ is a
$k$-bridge, and we have $n_e(\s) = 1$, otherwise. Since the algorithm processes
the $k$-bridges in reverse order, we have that $f^{a+1}_e(n_e(\s))$ is defined
for every edge $e \in p_4$. The following lemma shows that
line~\ref{l:bridgesum} correctly computes the cost of $b_j$.

\begin{lemma}
\label{LEM:BRIDGECORRECT}
Let $k$ be a vertex, and let $f^a$ be a partial equivalent cost function, such
that $f^{a}_e(j)$ is defined for every $j \le i$. Algorithm~\ref{alg:bridges}
computes a partial equivalent cost function $f^{a+1}$, such that $f^{a+1}$ is an
extension of $f^a$, and $f^{a+1}_e$ is defined for every $e$ that is a
$k$-bridge.
\end{lemma}
\begin{proof}
It can be verified that the algorithm constructs a partial cost function
$f^{a+1}$ that is an extension of~$f^{a}$, where~$f^{a+1}_e$ is defined for
every $e$ that is a $k$-bridge. We must show that~$f^{a+1}$ is partially
equivalent to~$f$. Since~$f^{a}$ is partially equivalent to $f$, there exists
some total cost function $f'$ that is an extension of $f^{a}$, such that $f'$ is
equivalent to $f$. We will show that $f'$ is also an extension of $f^{a+1}$.

We will do so inductively. The inductive hypothesis is that $f^{a+1}_e(i+1) =
f'_e(i+1)$ for every $e = b_l$ with $\ell > j$. The base case, where $j = m$, is
trivial, because there are no $k$-bridges $b_l$ with $\ell > m$. Now suppose that
we have shown the inductive hypothesis for some $j$. We show that
$f^{a+1}_{b_j}(i+1) = f'_{b_j}(i+1)$. Let $\s$ be the strategy queried when the
algorithm considers $b_j$. 

Consider an edge $e \in p_1$. By Lemma~\ref{LEM:P1P3}, we have that $n_e(\s) =
1$. By assumption, we have that $f^{a+1}_e(1) = f^{a}_e(1)$ for every edge $e$,
and therefore $f^{a+1}_e(n_e(\s)) = f'_e(n_e(\s))$ for every edge $e \in p_1$.

Now consider an edge $e \in p_4$. By Lemma~\ref{LEM:P4P5}, we
have that $n_e(\s) = 1$ whenever $e$ is not a $k$-bridge, and we have $n_e(\s) =
i+1$ whenever $e$ is a $k$-bridge. Therefore, by the inductive hypothesis, we
have that $f^{a+1}_e(n_e(\s)) = f'_e(n_e(\s))$ for every $e \in p_4$. 

Since $f'$ is equivalent to $f$, we have that $c_\s(p_1 b_l p_3) = \sum_{e \in
p_1 b_l p_3} f'_e$. Therefore, line~\ref{l:bridgesum} sets:
\begin{align*}
f^{a+1}_{b_j}(i+1) &= c_\s(p_1 b_j p_4) - \sum_{e \in p_1} f^{a+1}_e(n_e(\s)) - \sum_{e \in p_4} f^{a+1}_e(n_e(\s)) \\
&= \sum_{e \in p_1 b_j p_4} f'_e(n_e(\s)) - \sum_{e \in p_1} f'_e(n_e(\s)) -
\sum_{e \in p_4} f'_e(n_e(\s)) \\
&= f'_{b_j}(n_e(\s)) = f'_{b_j}(i + 1).
\end{align*}
Thus, the algorithm correctly sets $f^{a+1}_{b_j}(i+1) = f'_{b_j}(i+1)$.
\end{proof}

\paragraph{Incoming edges of $k$.}

We now describe the second part of the many-player case. After finding
the cost of each $k$-bridge, we find the cost of each incoming edge at $k$.
The following diagram describes
how we find the cost of $e = (v, k)$, an incoming edge at $k$ .
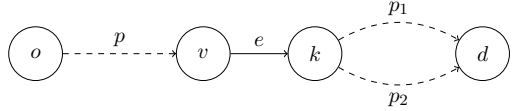
\begin{center}
\resizebox{0.45\textwidth}{!}{
\begin{tikzpicture}[node distance=3cm]
\node[state] (o) {$o$};
\node[state] [right of=o] (v) {$v$};
\node[state] [right of=v,node distance=2cm] (k) {$k$};
\node[state] [right of=k] (d) {$d$};
\path[->]
	(v) edge node[auto] {$e$} (k)
	(o) edge [dashed] node[auto] {$p$} (v)
	(k) edge [dashed,bend left] node[auto] {$p_1$} (d)
	(k) edge [dashed,bend right] node[auto,swap] {$p_2$} (d)
	;
\end{tikzpicture}
}
\end{center}
The path $p$ is an arbitrarily chosen path from $o$ to $v$. The paths $p_1$ and
$p_2$ are chosen according to the following lemma.

\begin{lemma}
\label{LEM:KDISJOINT}
There exist two $k$-$d$ paths $p_1, p_2$ such that every edge in $p_1 \cap
p_2$ is a $k$-bridge.
\end{lemma}
\begin{proof}
Let $b_1$ be the first $k$-bridge. By Lemma~\ref{LEM:MENGER} there exists edge
disjoint paths from $k$ to $b_1$. The proof can then be completed by applying 
Lemma~\ref{LEM:P4P5}. 
\end{proof}

\begin{algorithm}
\begin{algorithmic}[1]

\Require A vertex $k$, and a partial equivalent cost function $f^a$, such that
$f^{a}_e(j)$ is defined for all $e \in E$ when $j \le i$, all $e = (v,
u)$ with $u \prec k$ when $j = i+1$, and all $k$-bridges when $j = i+1$.
\Ensure A partial equivalent cost function $f^a$, such that
$f^{a}_e(j)$ is defined for all $e \in E$ when $j \le i$, and for all $e = (v,
u)$ with $u \preceq k$ when $j = i+1$.

\ForAll {$e$ and $j$ for which $f^{a}_e(j)$ is defined}
\label{l:copystart4}
\State $f^{a+1}_e(j) \leftarrow f^{a}_e(j)$
\EndFor
\label{l:copyend4}

\ForAll {$e = (v, k) \in E$}
\State $p \leftarrow $ an arbitrary $o$-$v$ path
\State $p_1, p_2$ paths chosen according to Lemma~\ref{LEM:KDISJOINT}
\State $\s \leftarrow (1 \mapsto p e p_1, i \mapsto p e p_2)$
\State $c_\s \leftarrow \query(\s)$
\State $f^{a+1}_e(i+1) \leftarrow c_\s(p e p_1) - \sum_{e' \in p}
f^{a+1}_{e'}(i+1) - \sum_{e' \in p_1} f^{a+1}_{e'}(n_{e'}(\s)).$
\label{l:manysum}
\EndFor
\end{algorithmic}
\caption{\textsc{MultiProcessK}}
\label{alg:manyk}
\end{algorithm}

Algorithm~\ref{alg:manyk} shows how we find the cost of putting $i + 1$ players
on each edge~$e$ that is incoming at~$k$. Apart from the consideration of
$k$-bridges, this algorithm uses the same technique as Algorithm~\ref{alg:k}.
Consider line~\ref{l:manysum}. Note that every vertex in $p$ is processed before
$k$ is processed, and therefore $f^{a+1}_{e'}(i+1)$ is known for every $e' \in
p$. Moreover, for every edge $e' \in p_1$, we have that $n_{e'}(\s) = i+1$ if
$e'$ is a $k$-bridge, and we have $n_{e'}(\s) = 1$ otherwise. In either case,
the $f^{a+1}_{e'}(n_{e'}(\s))$ is known for every edge $e' \in p_1$. The
following lemma show that line~\ref{l:manysum} correctly computes
$f^{a+1}_e(i+1)$.

\begin{lemma}
\label{LEM:ALGMANYK}
Let $k$ be a vertex, and let $f^a$ be a partial equivalent cost function, such
that $f^{a}_e(j)$ is defined for all $e \in E$ when $j \le i$, all $e =
(v, u)$ with $u \prec k$ when $j = i+1$, and all $k$-bridges when $j = i+1$.
Algorithm~\ref{alg:manyk} produces a partial equivalent cost function $f^{a+1}$,
such that $f^{a+1}_e(j)$ is defined for all $e \in E$ when $j \le i$, and for
all $e = (v, u)$ with $u \preceq k$ when $j = i+1$.
\end{lemma}
\begin{proof}
It can be verified that the algorithm constructs a partial cost function
$f^{a+1}$ that is defined for the correct parameters. We must show that
$f^{a+1}$ is partially equivalent to~$f$. Note that $f^{a+1}$ is an extension of
$f^a$. Since $f^{a}$ is partially equivalent to $f$, there exists some total
cost function $f'$ that is an extension of $f^{a}$, such that $f'$ is equivalent
to $f$. We will show that $f'$ is also an extension of $f^{a+1}$.

Let $e = (v, k)$ be an incoming edge at $k$. We will show that $f^{a+1}_e(i+1) =
f'_e(i+1)$. Let $\s$ be the strategy that the algorithm queries while processing
$e$. Since $f'$ is equivalent to $f$, we have that $c_\s(p e p_1) = \sum_{e' \in
p e p_1} f'_{e'}(n_{e'}(\s))$. For every edge $e' \in p_1$, we have $n_{e'}(\s) = i+1$.
Since every vertex $w$ visited by $p$ satisfies $w \prec k$, for every $e' \in
p_1$ we must have $f^{a+1}_{e'}(n_{e'}(\s)) = f^{a}_{e'}(n_{e'}(\s)) =
f'_{e'}(n_{e'}(\s))$. For
every edge $e' \in p_1$, we have $n_{e'}(\s) = 1$ if $e'$ is not a $k$-bridge, and we
have $n_{e'}(\s) = i+1$ if $e'$ is a $k$-bridge. In either case, we have that
$f^{a+1}_{e'}(n_{e'}(\s)) = f^{a}_{e'}(n_{e'}(\s)) = f'_{e'}(n_{e'}(\s))$ for
every edge $e' \in p_1$. Therefore, line~\ref{l:manysum} sets:
\begin{align*}
f^{a+1}_e(i+1) &= c_\s(p e p_1) - \sum_{e' \in p} f^{a+1}_{e'}(n_{e'}(\s)) -
\sum_{e' \in p_1} f^{a+1}_{e'}(n_{e'}(\s)) \\
&= \sum_{e \in p e p_1} f'_{e'}(n_{e'}(\s)) - \sum_{e' \in p} f'_{e'}(n_{e'}(\s)) - \sum_{e'
\in p_1} f'_{e'}(n_{e'}(\s)) \\
&= f'_e(n_e(\s)) = f'_e(i+1).
\end{align*}
Therefore, for each incoming edge $e = (v, k)$, we have that $f^{a+1}_e(i+1) =
f'_e(i+1)$. Hence, $f'$ is an extension of $f^{a+1}$, which implies that
$f^{a+1}$ is partially equivalent to $f$.
\end{proof}

\paragraph{Query complexity.}

We argue that the algorithm can be implemented so that the costs for $(i+1)$
players can be discovered using at most $|E|$ many payoff queries. 
Every time Algorithm~\ref{alg:bridges} discovers the cost of
placing $i+1$ players on a $k$-bridge, it makes exactly one payoff query. 
Every time Algorithm~\ref{alg:manyk} discovers the cost
of an incoming edge $(v, k)$, it makes exactly one payoff query. 
The key observation is that the costs discovered by Algorithm~\ref{alg:bridges}
do not need to be rediscovered by Algorithm~\ref{alg:manyk}. That is, we can
modify Algorithm~\ref{alg:manyk} so that it ignores every incoming edge $(v, k)$
that has already been processed by Algorithm~\ref{alg:bridges}. This
modification ensures that the algorithm uses precisely $|E|$ payoff queries to
discover the edge costs for $i + 1$ players. This gives us the following
theorem.

\begin{theorem}\label{thm:UBnetwork}
Let $\Gamma$ be a symmetric network congestion game with $n$-players played on a
DAG with $|E|$ edges. The payoff query complexity of finding a Nash equilibrium
in $\Gamma$ is at most $n \cdot |E|$.
\end{theorem}

\section{Conclusions and further work}
\label{sec:fw}



We first consider open questions in the setting of payoff queries, which has
been the main setting for the results in this paper. We then consider alternative
query models.

\paragraph{Open questions concerning payoff queries.}


In the context of strategic-form games, there are a number of open problems.
In Theorem~\ref{thm:lowerlog}, we show a super-linear lower bound on the payoff
query complexity when $\epsilon$ is allowed to depend on $k$.
Can we prove a super-linear lower bound for a constant $\epsilon$?
Is there a deterministic algorithm that can find an $\epsilon$-Nash equilibrium
with $\epsilon<\frac{1}{2}$ without querying the entire payoff matrices?
\citet{FS13} achieve $\epsilon<\frac{1}{2}$ with the use of randomization, but
doing so with a deterministic algorithm appears to be challenging.
Finally, when $2 \le i \le k-1$, we have
shown that the payoff query complexity of finding a $(1 - \frac{1}{i})$-Nash
equilibrium lies somewhere in the range $[k-i+1, 2k - i +1]$. Determining the
precise payoff query complexity for this case is an open problem.

For congestion games, our lower bound of $\log n + m$ arises
from a game with two parallel links and a one-player game with
$m$ links.
The upper bound of 
$\bigO\left(\log(n)\cdot \frac{\log^2(m)}{\log\log(m)}+m\right)$
is a poly-logarithmic factor off from this lower bound,
with the factor depending on $m$.
Can this factor be improved?
It seems unlikely that the dependence of this factor on $m$ can be completely
removed, in which case, in order to provide tight bounds, a single lower bound
construction that depends simultaneously on $n$ and $m$ would be necessary.
 
For symmetric network congestion games on DAGs it is unclear whether the payoff
query complexity is sub-linear in $n$. Non-trivial lower and upper bounds for
more general settings, such as asymmetric network congestion games (DAG or not)
or general (non-network) congestion games would also be interesting.




\paragraph{Other query models.} 

We have defined a payoff query as given by a {\em pure} (not mixed) profile
$\s$, since that is of main relevance to empirical game-theoretic modelling.
Furthermore, if $\s$ was a mixed profile, it could be simulated by sampling a
number of pure profiles from $\s$ and making the corresponding sequence of pure
payoff queries. An alternative definition might require a payoff query to just
report a single specified player's payoff, but that would change the query
complexity by a factor at most $\numplayers$.


Our main results have related to exact payoff queries, though other query
models are interesting too.
%
%
%
%
%
%
%
A very natural type of query is a {\em best-response query}, where a strategy
$\s$ is chosen, and the algorithm is told the players' best responses to $\s$. 
In general $\s$ may have to be a mixed strategy; it is not hard to check that
pure-strategy best response queries are insufficient; even for a two-player
two-action game, knowledge of the best responses to pure profiles is not
sufficient to identify an $\epsilon$-Nash equilibrium for
$\epsilon<\frac{1}{2}$.
Fictitious Play (\citet{FL}, Chapter 2) can be regarded as a query protocol that
uses best-response queries (to mixed strategies) to find a Nash equilibrium in
zero-sum games, and essentially a 1/2-Nash equilibrium in general-sum
games~\citep{GSSV13}.
%
We can always synthesize a pure best-response query with $n(k-1)$ payoff queries.
Hence, for questions of polynomial query complexity, payoff queries are at least as
powerful as best-response queries.
%
%
Are there games where best-response queries are much more useful than payoff queries?
%
If $\numstrats$ is
large then it is expensive to synthesize best-response queries with payoff queries.
The DMP-algorithm~\citep{DMP} finds a $\frac{1}{2}$-Nash equilibrium via only
two best-response queries, whereas Theorem~\ref{thm:2playerpayoffs} notes that
$\bigO(\numstrats)$ payoff queries are needed.



A {\em noisy} payoff query outputs an observation of a random variable taking values
in $[0,1]$ whose expected value is the true payoff. Alternative versions might assume
that the observed payoff is within some distance $\epsilon$ from the true payoff.
Noisy query models might be more realistic, and they are suggested by
by the experimental papers on querying games.
However in a theoretical context, one could obtain good approximations of
the expected payoffs for a profile $\s$, by repeated sampling.
%
It would interesting to understand the power of different query models. 

\begin{acks}
We would like to thank Michael Wellman for interesting discussions on this
topic. This work was supported by ESRC grant ESRC/BSB/09, and EPSRC grants
EP/K01000X/1, EP/J019399/1, EP/H046623/1, and EP/L011018/1.
\end{acks}

\bibliography{references}

\end{document}